\documentclass[11pt, letterpaper]{article}

\usepackage{comment}
\usepackage{tabularx}
\usepackage{graphicx}

\usepackage{amsmath, amssymb, mathtools}
\usepackage{amsthm}
\usepackage{bm}
\usepackage{bbm}
\allowdisplaybreaks
\usepackage[labelfont=bf]{caption}

\usepackage[T1]{fontenc}
\usepackage[scaled=0.8]{beramono}

\usepackage{xcolor}
\usepackage{colortbl}
\definecolor{linkcolor}{HTML}{0645AD}
\definecolor{crimson}{HTML}{A41034}
\definecolor{fowleremrobustrow}{HTML}{EAF6EF}

\usepackage{hyperref}
\hypersetup{
     colorlinks   = true,
     allcolors    = blue, 
}
\usepackage{float}

\def\*#1{\mathbf{#1}}

\newcommand{\E}{\mathbb{E}}
\newcommand{\bbone}{\mathbbm{1}}

\newcommand{\V}{\mathbb{V}}

\newcommand{\indep}{\!\perp\!\!\!\perp}

\newtheorem{theorem}{Theorem}
\newtheorem{proposition}{Proposition}
\newtheorem{lemma}{Lemma}
\newtheorem{assumption}{Assumption}
\newtheorem{remark}{Remark}

\newtheorem{result}{Result}

\theoremstyle{definition}

\numberwithin{equation}{section}
\numberwithin{theorem}{section}
\numberwithin{proposition}{section}
\numberwithin{lemma}{section}
\numberwithin{assumption}{section}
\numberwithin{example}{section}

\usepackage{setspace}
\usepackage{pdfpages}
\usepackage{booktabs}
\usepackage{multirow}
\usepackage{enumitem}

\usepackage{algorithm}
\usepackage{algpseudocode}

\usepackage{tikz}
\usetikzlibrary{arrows.meta, positioning}

\usepackage[round, authoryear]{natbib}
\title{Debiased Inference for AI-Generated Data \\ without Gold-Standard Labels: \\ Identification via Multiple Imperfect Measurements\thanks{We thank Betsy Ogburn, Helen Guo, AmirEmad Ghassami, and Ilya Shpitser for thoughtful comments.}}
\author{Naoki Egami\thanks{Associate Professor, Department of Political Science and Statistics and Data Science Center, Massachusetts Institute of Technology, Cambridge, MA 02139.
Email: \href{mailto:egami@mit.edu}{egami@mit.edu}
URL: \href{https://naokiegami.com/}{https://naokiegami.com}}
\and 
Sooahn Shin\thanks{Postdoctoral Associate, 
Department of Political Science, 
Massachusetts Institute of Technology, Cambridge, MA 02139.
Email: \href{mailto:sshin3@mit.edu}{sshin3@mit.edu}
URL: \href{https://sooahnshin.com}{https://sooahnshin.com}}
}
\date{}

\begin{document}
\maketitle

\begin{abstract}
An increasing number of scholars use AI to measure variables they subsequently include in downstream analyses. Although AI-measured variables are often analyzed as if observed without error, ignoring prediction errors in automated measurement leads to substantial bias and invalid confidence intervals in downstream analyses, even if AI measurement accuracy is high, e.g., above 90\%. Existing solutions, such as design-based supervised learning and prediction-powered inference, combine error-prone AI-based measurements with gold-standard labels, which may be costly and difficult to obtain in some application areas.

In this paper, we propose \textit{debiased inference with multiple imperfect measurements} (DMM), a framework that combines multiple error-prone AI measurements to enable valid downstream inference without gold-standard labels. Building on the established results on CP decomposition, DMM assumes that these measurements are independent conditional on the latent true label and observed unit-level features, such as text features represented by embeddings. This framework allows for unknown misclassification rates to vary across annotation methods (e.g., large language models) and across units of annotation (e.g., texts). Under this assumption, we use semiparametric inference theory to prove that the DMM estimator is consistent and asymptotically normal, enabling valid inference for a wide range of downstream statistical analyses common in the social sciences. Our simulation results show that DMM yields valid inference and that adding accurate, though imperfect, measurements can improve efficiency. Focusing on common applications of large language model annotations, we also develop diagnostics to assess the conditional independence assumption.
\end{abstract}

\clearpage
\section{Introduction}

One of the most common applications of AI in the social sciences has been in measurement. Researchers have used AI to measure a wide range of variables, such as sentiment, tone, topics \citep[e.g.,][]{gilardi2023chatgpt, ziems2024can}, protests, political violence \citep[e.g.,][]{halterman2026codebook}, and job types \citep[e.g.,][]{hansen2023remote}, among others. Such measurement tasks are only the first step. Scholars often use learned measurements as key variables of interest in downstream analyses. For example, \citet{pan2018concealing} annotate whether each online post accuses local Chinese officials of corruption so that they can later study whether and how much such online complaints are censored. \citet{fowler2021political} annotate the tone of political ads and then analyze how politicians strategically change the tone of political advertising online and offline.

When using such AI-measured variables in downstream analyses, it might be tempting to ignore measurement errors---the unknown and heterogeneous mismatch between the gold-standard label and the AI measurements---and analyze such variables as if they were observed without errors. However, recent papers theoretically and empirically demonstrate that ignoring errors in the first-stage measurement step can lead to substantial bias and invalid confidence intervals, even if the accuracy of the AI measurement is high, e.g., 90\%. This is because such measurement errors are non-random and nonclassical and correlated with observed and unobserved variables that matter in downstream inference \citep{schennach2016recent, wang2020methods, egami2023using}. In practice, this means that researchers can get substantively and statistically different results if they choose slightly different measurement methods, such as large language model (LLM) annotations with different prompts, temperature, and so on \citep[e.g.,][]{baumann2025large, yang2026benchmark}. These biases can also limit the replicability of research based on AI-based measurements \citep{spirling2023open}, as many AI models that are currently the state of the art will be deprecated in the near future. 

A popular existing solution is to combine a small number of gold-standard labels with error-prone AI-based measures via a doubly robust procedure \citep{robins1994, cher2018double}. Widely used variants include design-based supervised learning \citep[DSL;][]{egami2023using}, prediction-powered inference \citep[PPI;][]{angelopoulos2023prediction}, and MAR-S \citep{carlson2025unifying}. They assume that researchers control the process through which each unit is sampled for gold-standard labeling (often simple random sampling but can also accommodate unequal sampling based on observed features). The key strength is that these methods make no assumptions about errors in AI-based measures. However, the biggest requirement is that they need high-quality gold-standard labels as the validation data, which might be costly and difficult to obtain in some application areas. 

In this paper, we propose a framework of \textit{debiased inference with multiple imperfect measurements} (DMM) that combines three or more imperfect measurements in place of gold-standard labels to conduct valid downstream inference with AI-generated data. For identification, we build on the long-standing literature on nonparametric latent variable models and CP decomposition \citep{kruskal1977three, hu2008identification, allman2009identifiability}. The classical conditional independence assumption in this literature \citep{dawid1979maximum} assumes that multiple imperfect measurements, also known as proxies, are independent conditional on the true latent variable. This assumption is weaker than assuming that proxies have no measurement error and allows misclassification rates to vary across proxies. However, it does not allow for any shared source of errors (e.g., some documents might be more difficult to annotate for every LLM; \citealp{chen2026partial}). We instead assume that multiple proxies are independent conditional not only on the unobserved true labels but also on any observed features of each annotation task, such as the complexity and difficulty of texts captured by rich text embeddings. This assumption allows for unknown misclassification rates to vary not only across measurement methods (e.g., different LLMs) but also across units (e.g., texts). Our approach allows proxies to share sources of errors as long as such common causes are captured by conditioning variables, which can include rich text representation, for example. Importantly, if researchers wish to avoid assuming access to gold-standard labels, they have to make some assumptions about proxies, and our paper makes this assumption explicit and transparent. Recognizing its importance, Section~\ref{sec:conditional-independence-practice} offers a series of strategies to make the conditional independence assumption more plausible and to statistically evaluate the observational implications of this assumption. 

For estimation and inference, the proposed DMM estimator builds on semiparametric causal inference, in particular \citet{zhou2024causal} and \citet{guo2026proximal}, to perform debiased downstream inference. Simply estimating the finite mixture model (or variants of the Dawid-Skene model \citep{dawid1979maximum} or Kruskal decomposition \citep{kruskal1977three, allman2009identifiability}) and including the learned latent variables directly in downstream inference is not sufficient. We need to use a tailored debiased moment estimator to achieve Neyman orthogonality \citep{cher2018double} and allow for slow convergence rates of the estimation of conditional classification rate models, which are nuisance functions for downstream inference. The proposed DMM estimator is consistent and asymptotically normal, and its corresponding
confidence interval is valid, under the conditional independence assumption and mild assumptions about convergence rates of the nuisance functions. By extending the existing literature, we allow for (a) a wide range of downstream analyses common in the social and biomedical sciences that can be written as the moment estimator (e.g., most maximum likelihood estimators); (b) settings where an error-prone variable is either a dependent or independent categorical variable in downstream analyses, and (c) more than three imperfect measurements (and their implications for efficiency). 

Thinking broadly, our proposed method and existing approaches based on gold-standard labels are complementary. The proposed DMM makes no assumption about the gold-standard labels but makes stronger assumptions about measurement errors in AI, whereas methods using validation data, such as DSL \citep{egami2023using} and PPI \citep{angelopoulos2023prediction}, make no assumption about measurement errors in AI but assume access to gold-standard labels. Which method is more appropriate depends on applications, and we discuss when and how to combine these two methods. 
\\ 

Our contributions can be summarized in three points. 
\begin{enumerate}
\item \textbf{Exploiting Multiple Imperfect Measurements in Place of Gold-Standard Labels}: We develop a method that combines multiple proxies to conduct valid downstream inference for AI-generated data, without assuming access to gold-standard labels. This is in contrast to existing bias-correction methods, such as DSL and PPI, that rely on gold-standard data. 
\item \textbf{General Applicability}: We extend existing theoretical results by allowing for (a) a wide range of downstream analyses common in practice, including linear regression, logistic regression, and most maximum likelihood estimators, (b) settings where an error-prone variable is either a dependent or independent categorical variable in downstream analyses, and (c) more than three proxies (e.g., more than three annotation methods). This general applicability is fundamental to support diverse downstream analyses conducted in empirical sciences. 
\item \textbf{Application to LLM annotations}: While the method itself is applicable for downstream inference with any error-prone variables, we specifically focus on the most common applications of LLM annotations for AI-generated data. In particular, we discuss (a) how to make the conditional independence assumption more plausible by carefully choosing conditioning variables (e.g., estimated task difficulty and text embeddings), proxies (e.g., different families of LLMs), and prompts (e.g., randomly selecting different prompts for each LLM); (b) how to use more than three proxies to perform the overidentification test of the conditional independence assumption, and (c) how to combine DMM and DSL in a special case where some gold-standard labels are also available. 
\end{enumerate}

After describing related work, we formally characterize the problem setting and existing methods (Section~\ref{sec:problem-setting}). In Section~\ref{sec:latent-independent}, we describe our proposed method and prove its theoretical properties. Here, for the sake of clear presentation, we focus on settings where an error-prone measurement is an independent variable in downstream analysis. We start with nonparametric identification and then derive asymptotic statistical properties of the DMM estimator in estimation and inference. In Section~\ref{sec:latent-dependent}, we extend the framework to settings in which the annotated construct is a binary dependent variable and then discuss the most general case where a multicategory label can be either a dependent or independent variable. In Section~\ref{sec:conditional-independence-practice}, we come back to the core assumption of conditional independence and discuss how to make the assumption more plausible in practice. We also develop a series of statistical diagnostic tools to assess the conditional independence assumption. In Section~\ref{sec:applications}, we provide extensive simulation and empirical validation studies to demonstrate the statistical properties of DMM. Section~\ref{sec:discussion} concludes with a discussion.

\subsection*{Related literature}

\paragraph{AI-Generated Data.}
With recent advances in AI and LLMs, a growing literature develops methods for valid downstream inference using AI- or machine-learning-generated measurements together with a gold-standard validation sample.
Popular examples include design-based supervised learning \citep{egami2023using,egami2024using}, prediction-powered inference \citep{angelopoulos2023prediction}, methods reviewed in \cite{ludwig2024large}, MAR-S \citep{carlson2025unifying}, model-assisted impact analysis \citep{mozer2023decreasing}, and related control-variate approaches \citep{katsumata2023statistical}. They combine predictions available at scale with a smaller number of gold-standard labels to bias-correct downstream inference. All of these methods build on the longstanding literature on semiparametric inference and causal inference \citep[e.g.,][]{robins1994, chen2000unified, chen2008semiparametric, cher2018double}.
These methods assume access to gold-standard labels and knowledge of their sampling design, whereas DMM does not require such gold-standard labels.

There are several recent papers that relax the assumption of gold-standard labels. \citet{battaglia2024inference} assume a new asymptotic regime where measurement errors and sampling errors are comparable and decrease as sample size grows. They use partial validation data where one only observes the gold-standard labels and error-prone AI measurements instead of the full validation where researchers need to observe not only the gold-standard labels and AI measurements but also all the downstream variables for each unit. \citet{chen2026partial} derive bounds on latent label prevalence and regression coefficients from panels of LLM reports under externally calibrated score and event restrictions (e.g., reporter-specific accuracy restrictions), while allowing arbitrary dependence across reports conditional on the latent truth. 

DMM complements these approaches
by identifying the latent variable through the established array decomposition of multiple imperfect measurements. Compared to the classical conditional independence assumption \citep[e.g.,][]{dawid1979maximum} that only conditions on the true unobserved label and does not allow for any shared source of errors across multiple measurements, our conditional independence assumption explicitly allows for conditioning on a rich set of covariates that may contain observed or derived characteristics of the input (e.g., writing style, text length, language, or image quality) measured by rich text embeddings. Recognizing the importance of the assumption, we also develop a series of statistical diagnostics in Section~\ref{sec:conditional-independence-practice}. We also find that DMM performs well in our empirical validation study in Section~\ref{sec:applications} where the conditional independence assumption holds only approximately.

\paragraph{Repeated Measurement Identification.}
Our work also builds on the literature on identifying latent variables from conditionally independent repeated measurements. A large classical measurement error literature assumes an additive error that is independent of the latent variable and mean zero, or at least mean zero conditional on it, and obtains identification using instrumental variables or repeated measurement identities
\citep{schennach2016recent}.
Such restrictions are ill-suited to categorical labels:
misclassification of a finite-support variable is inherently nonclassical, and for nominal categories an additive error representation is itself unnatural \citep[Section~6.1]{schennach2016recent}.
Accordingly, our identification strategy instead builds on the nonclassical measurement error literature.

In foundational work, \citet{kruskal1977three} studies three-way array decomposition and provides rank conditions under which the decomposition is unique up to a common permutation and scaling of its latent components. 
This result underlies the identification of discrete latent structures from conditionally independent measurements. 
\citet{hu2008identification} studies nonlinear models with a misclassified discrete independent variable,
\citet{hu2008instrumental} develop related identification results for continuous variables, and 
\citet{allman2009identifiability} establish related results for discrete models with hidden variables.
\citet{chen2011nonlinear} and \citet{schennach2016recent, schennach2022measurement} provide a broader and unifying review of this literature, including such identification strategies based on multiple imperfect measurements under classical and nonclassical measurement error. In contrast to this literature, we allow for more than three imperfect measurements, a more general class of downstream inference, and semiparametric estimation and inference methods that are Neyman orthogonal to the first-stage nuisance function estimation. 

The Dawid--Skene model uses the same core conditional-independence structure:
it treats the true item label as latent, allows each annotator to have a distinct confusion matrix, and assumes that annotations are independent conditional on the true label \citep{dawid1979maximum}. 
In the context of LLM annotations, \citet{bouyamourn2026interpretable} characterize conditions under which correlated annotations can be aggregated and propose a dependence-aware extension of the Dawid--Skene model. 

\paragraph{Proxy-based Causal Inference.}
Another related literature studies causal inference with unmeasured variables using proxies \citep{kuroki2014measurement, miao2018identifying, egami2024identification}, in particular, \citet{zhou2024causal, guo2026proximal}. Most importantly, \cite{zhou2024causal} develop a method to make causal inference with a latent treatment. Specifically, they also build on the nonparametric latent variable model literature \citep{kruskal1977three, allman2009identifiability} for nonparametric identification and derive a new semiparametric efficiency theory and semiparametric estimation strategies for causal effects with a hidden treatment. \citet{guo2026proximal} consider causal inference with the latent outcome variable and derive an influence function-based semiparametric estimator. 
Their work considers settings where the latent outcome has arbitrary finite support and the proxies may be either discrete or continuous, and establishes general existence results for the constructions.
Similarly, \citet{nakamura2025surrogate} applies the array-decomposition ideas related to \citet{zhou2024causal} to text and image annotations with a latent outcome. Theoretically, we extend this literature by allowing for (a) more than three proxies (with analytical results of how increasing the number of proxies affects efficiency); (b) a general class of downstream analyses (in contrast to causal effects these previous studies have focused on), and (c) settings where a latent variable is either a dependent or independent categorical variable (each of previous studies focuses on just one of them).

\section{The Problem Setting and Existing Approaches}
\label{sec:problem-setting}

Scholars often use machine or human annotations to measure key variables of interest that they want to analyze in downstream analysis. However, as recent papers theoretically and empirically show, ignoring measurement errors in the annotation step can bias downstream inference, even when the accuracy of the annotation step is high, e.g., more than 90\%. For clarity of presentation, we first focus on settings where an error-prone variable is a binary independent variable in downstream analysis. In Section~\ref{sec:latent-dependent}, we generalize our method to settings where an error-prone variable is a general categorical variable that is either an independent or dependent variable.  

\subsection{Setup}
For each unit $i=1,\ldots,n$, let $X_i^\ast\in\{0,1\}$ denote an unobserved true label, $Y_i$ denote an observed downstream dependent variable, and $W_i$ denote observed covariates in downstream analysis. 
Instead of observing the true label $X_i^\ast$, researchers observe $J\geq 3$ error-prone proxy labels
\begin{equation*}
    \widetilde X_i
    =
    \bigl(X_i^{(1)},\ldots,X_i^{(J)}\bigr),
    \qquad X_i^{(j)}\in\{0,1\}.
\end{equation*}
The labels may, for example, be produced by different human annotators, machine learning models, or large language models (LLMs).
They need not have the same accuracy, and their errors may be systematic or heterogeneous across units. Especially when researchers use LLMs to measure $X^\ast$, it is often easy to obtain multiple imperfect proxies, as each proxy can come from different LLMs. 

Let $\psi^F(y,x,w;\beta) \in\mathbb R^{d_\beta}$
be a user-specified full-data moment function for downstream analysis, where
$\beta\in\mathcal B\subset\mathbb R^{d_\beta}$ denotes the finite-dimensional parameter of interest. 
If $X^\ast$ were observed for every unit, the oracle downstream parameter $\beta^\ast$ would be the unique solution to
\begin{equation}
    \E\!\left[
        \psi^F(Y,X^\ast,W;\beta)
    \right]
    =
    0.
    \label{eq:full-data-target}
\end{equation}

The moment estimator above can accommodate a large class of estimators as special cases, including most maximum likelihood estimators. For linear regression, 
\begin{equation*}
    \psi_{\text{lin}}^F(y,x,w;\beta)
    =
    (1, x,w^\top)^\top
    \left[
        y-(1, x, w^\top)\beta
    \right]
    \label{eq:lm-full-data-score}
\end{equation*}
and the coefficients of interest $\beta$ are the coefficients of the linear regression model that regresses $Y$ on $X^\ast$ and $W$. When you have logistic regression, a moment function is 
\begin{equation*}
    \psi_{\text{logistic}}^F(y,x,w;\beta)
    =
    (1, x,w^\top)^\top
    \left[
        y-\operatorname{expit}\{(1, x, w^\top)\beta\}
    \right].
    \label{eq:logistic-full-data-score}
\end{equation*}
A moment function for generalized linear models can be written as 
\begin{equation}
    \psi_{\text{glm}}^F(y,x,w;\beta)
    =
    (1, x,w^\top)^\top
    \left[
        y-\ell\{(1, x, w^\top)\beta\}
    \right],
    \label{eq:glm-full-data-score}
\end{equation}
where $\ell(\cdot)$ is an inverse link function specific to each model. Our theory and methods do not assume the downstream regression model chosen by users is correctly specified. Under misspecification, coefficients $\beta$ defined in equation~\eqref{eq:full-data-target} can be interpreted as the population projection parameter \citep[e.g.,][]{buja2019models, vansteelandt2022assumption}. 

\subsection{Existing Approaches}

\subsubsection{Naive Approaches Ignoring Measurement Errors}
A common strategy first collapses the multiple labels to a single proxy, e.g., using majority voting $\check X_i=
\mathbf{1}\{\sum_j X_i^{(j)}>J/2\}$. More generally, $\check X_i \in \{0,1\}$ may be any function of the multiple labels, such as selecting one annotator or thresholding their average. The resulting $\check X_i$ is then treated as if it were $X_i^\ast$. 
The naive estimator $\widehat\beta_{\mathrm{naive}}$ solves
\begin{equation*} 
    \frac{1}{n}\sum_{i=1}^n \psi^F(Y_i,\check X_i,W_i;\beta) = 0. 
\end{equation*}
This is exactly the downstream moment function one would obtain when users directly include $\check X$ in downstream analysis. Its population moment bias is
\begin{align}
    &\mathbb E\!\left[
        \psi^F(Y,\check X,W;\beta)
        -\psi^F(Y,X^\ast,W;\beta)
    \right]
    \ = \
    \mathbb E\!\left[
        (\check X-X^\ast)
        \{\psi^F(Y,1, W;\beta)-\psi^F(Y, 0, W;\beta)\}
    \right].
    \label{eq:naive-bias}
\end{align}
High classification accuracy of $\check X$ does not guarantee that the right-hand side of equation~\eqref{eq:naive-bias} is zero. This is because classification errors are nonclassical and correlated with observed and unobserved variables relevant in downstream regression.  Combining multiple imperfect measurements into one index may reduce unit-level classification errors, but it does not by itself justify treating the aggregated label as error-free in downstream analysis.

\subsubsection{Gold-Standard-Only Estimation}
To address non-random, nonclassical measurement errors, the most dominant existing approach is to rely on gold-standard labels. Such methods assume that $X_i^\ast$ is observed for a small subset of units sampled according to a known sampling design. Let $R_i\in\{0,1\}$ indicate whether $X_i^\ast$ is observed, and define
the sampling probability $\rho_i \coloneq \Pr(R_i=1\mid Y_i,W_i,\widetilde X_i)\in(0,1]$.
We assume that the sampling probability is controlled by the researcher, so that $\rho_i$ is known and $R_i \indep X_i^\ast \mid Y_i,W_i,\widetilde X_i$. The most common, simple random sampling is the special case where the sampling probability is constant and $\rho_i = \rho$ for all $i$.
The gold-standard-only estimation (GSO) $\widehat\beta_{\mathrm{GSO}}$ uses only the
gold-standard observations and solves
\begin{equation*}
    \frac{1}{n}\sum_{i=1}^n
    \frac{R_i}{\rho_i}
    \psi^F(Y_i,X_i^\ast,W_i;\beta)
    =0.
\end{equation*}
This is equivalent to running downstream regression only using a subset of data that have gold-standard labels.  It is straightforward to show that this estimator is consistent for $\beta^\ast$ and asymptotically normal, allowing for valid statistical inference when the sampling probability is known \citep{van2000asymptotic}. However, GSO tends to be inefficient as it discards information in imperfect labels and all units without gold-standard labels. 

\subsubsection{Debiased Inference with Gold-Standard Labels} More recently, a broad class of methods combines predictions available at scale with a smaller sample of gold-standard measurements to bias-correct downstream analysis. This class of bias-correction methods includes design-based supervised learning \citep[DSL;][]{egami2023using,egami2024using}, 
prediction-powered inference \citep[PPI;][]{angelopoulos2023prediction}, methods reviewed in \cite{ludwig2024large}, MAR-S \citep{carlson2025unifying},
model-assisted impact analysis \citep{mozer2023decreasing}, and the control-variate approach of \citet{katsumata2023statistical}, among others. All of these methods build on the longstanding literature on semiparametric inference and causal inference \citep[e.g.,][]{robins1994, chen2000unified, chen2008semiparametric, cher2018double}.
Although these methods share a common prediction-and-correction strategy, they differ in their sampling designs, prediction-training procedures, and efficiency adjustments. Here, we use DSL as a representative benchmark.

DSL uses the gold-standard observations to train a supervised predictor
of the latent independent variable and then bias-correct prediction errors in the downstream moment function via a doubly robust procedure. 
Using cross-fitting, let $\widehat X_i=\widehat g_{-k(i)}(Y_i,W_i,\widetilde X_i) \in [0,1]$
denote an out-of-fold prediction of $X_i^\ast$ where $g_{-k(i)}$ is estimated using gold-standard data excluding the fold $k(i)$ that unit $i$ belongs to. The prediction model can naturally use all $J$ proxies jointly, together with the observed
outcome and covariates, to predict the true label $X^\ast$.
The DSL moment function is
\begin{equation}
    \psi_i^{\mathrm{DSL}}(\beta)
    \coloneq
    \psi^F(Y_i, \widehat{X}_i, W_i;\beta)
    +
    \frac{R_i}{\rho_i}
    \left\{
        \psi^F(Y_i,X_i^\ast,W_i;\beta)
        -
        \psi^F(Y_i, \widehat{X}_i, W_i;\beta)
    \right\}.
    \label{eq:dsl-moment}
\end{equation}
More generally, the DSL estimator can use any generic out-of-fold prediction $\widehat{m}_i(\beta)$ for $\psi^F(Y_i,X_i^\ast,W_i;\beta)$, and the expression above with $\widehat{m}_i(\beta) = \psi^F(Y_i, \widehat{X}_i, W_i;\beta)$ is a special case of the general DSL estimator. The DSL estimator $\widehat\beta_{\mathrm{DSL}}$ solves
$\frac{1}{n}\sum_{i=1}^n \psi_i^{\mathrm{DSL}}(\beta)=0$. This estimator makes a key assumption that the sampling probability $\rho_i$ is controlled by and known to researchers. This scenario is common in many application areas where, for example, researchers choose which documents to be coded by experts. Under the known sampling design,
\begin{equation*}
    \mathbb E\!\left[
        \psi_i^{\mathrm{DSL}}(\beta)
        \,\middle|\,
        X_i^\ast,Y_i,W_i,\widetilde X_i,\widehat X_i
    \right]
    \ = \ 
    \psi^F(Y_i,X_i^\ast,W_i;\beta).
\end{equation*}
Consequently, DSL targets the oracle parameter even when the supervised prediction model and LLM annotations are arbitrarily misspecified. These classes of methods are popular and powerful as they do not require any assumption about errors made by LLMs or any methods used to generate $\widehat{X}_i$. When $\widehat{X}_i$ is more accurate, DSL becomes more accurate too, while it is always valid without any assumption on $\widehat{X}_i$. 

These existing strategies span a continuum between relying on proxy labels and relying on gold-standard labels. Naive approaches use proxy labels as if they are error-free and they are in general biased unless proxy labels are perfect. GSO is unbiased and provides valid inference but ignores proxy labels and can be inefficient.
More recent bias-correction methods, such as DSL \citep{egami2023using,egami2024using} and PPI \citep{angelopoulos2023prediction}, provide valid inference under the known sampling probability for gold-standard labeling without making assumptions about errors in proxy variables. However, they assume access to gold-standard labels for a subset of units, which might be difficult or too costly for some applications. In the next section, we explore an alternative approach that does not assume access to gold-standard labels.

\section{Debiased Inference with Multiple Imperfect Measurements}
\label{sec:latent-independent}

We now show how to use multiple imperfect labels to perform valid downstream inference without requiring gold-standard labels. In particular, building on established results on identification of latent variables, we use the
joint distribution of multiple imperfect labels to recover the unbiased downstream moment under a conditional independence assumption. In Section~\ref{sec:latent-independent-assumptions}, we begin by reviewing this classical assumption and then relax it to accommodate modern applications of LLM annotations. We then discuss nonparametric identification (Section~\ref{subsec:iden}) and propose a consistent and asymptotically normal estimator to enable valid downstream inference (Section~\ref{sec:latent-independent-estimation}).

\subsection{Assumptions}
\label{sec:latent-independent-assumptions}

As is clear from the previous section, if researchers assume no access to gold-standard labels, they have to instead make some assumptions about proxy labels \citep{schennach2016recent}. Here, we build on the established literature of nonparametric latent variable models \citep[e.g.,][]{goodman1974exploratory, kruskal1977three, dawid1979maximum, hu2008identification, allman2009identifiability}. 
We begin with a brief review of the classical conditional independence assumption (see Section 6.1 of \citealp{schennach2016recent} for a detailed review). Specifically, the classical methods assume that multiple imperfect measurements are independent conditional on the true label: 
\begin{equation}
    X^{(1)}\indep\cdots\indep X^{(J)}
    \mid X^\ast. \label{eq:classical-cond}
\end{equation}
Importantly, the proxy labels may be biased, unequally accurate, and nonidentically distributed. 
In particular, their false-positive and false-negative rates may be nonzero and may differ across labels. This assumption is substantially weaker than assuming proxy labels are perfect and have no measurement errors. 

A large literature has developed methods for latent variable models under this conditional independence structure.
Examples include nonparametric estimation using an algorithm and joint approximate diagonalization \citep[][respectively]{hall2003nonparametric, bonhomme2016estimating},
crowdsourcing and noisy label models \citep{raykar2010learning, zhang2016spectral}, 
and tensor based methods or canonical correlation analysis for multi-view latent variable learning \citep[][the latter under a weaker conditional moment restriction]{anandkumar2014tensor, chaudhuri2009multi}.

However, this classical assumption may be violated if human annotators make similar errors because of shared characteristics of the input or annotation task (e.g., some texts are much longer and more difficult to annotate). Similarly, machine annotators may make correlated errors because of shared model characteristics that interact with the input (e.g., overlapping training data or a common prompt). For example, recent papers emphasize that errors in LLMs are correlated even conditional on the true label \citep{kim2025correlated, chen2026partial}. 

To make this assumption more plausible, we allow for conditioning on auxiliary input- or annotation-level information $D_i$, the downstream covariates $W_i$, and the downstream dependent variable $Y_i$.
The vector $D_i$ may account for heterogeneity in label accuracy or shared sources of dependence but need not enter the downstream regression. 

\begin{assumption}[Independence across labels conditioning on input and annotation information]
\label{ass:label-conditional-independence}
The multiple imperfect labels are mutually independent conditional on the latent label
and the observed covariates:
\begin{equation*}
    X^{(1)}\indep\cdots\indep X^{(J)}
    \mid X^\ast,\widetilde D,
\end{equation*}
where $\widetilde D = (D,W,Y).$ 
\end{assumption}

Assumption~\ref{ass:label-conditional-independence} allows for the misclassification rates to vary across annotators and across units. It specifically allows imperfect labels to share sources of errors as long as such common causes are captured by $\widetilde D$. The vector $D$ may contain observed or derived characteristics of the input (e.g., writing style, text length, language, or image quality) measured by rich text embeddings. It may also contain annotation-level information that varies across units (e.g., prompt version, task duration, annotation order,
or recorded fatigue). This assumption is significantly weaker than the classical conditional independence assumption stated in~\eqref{eq:classical-cond}, which did not allow for any shared source of errors. Assumption~\ref{ass:label-conditional-independence} still fails if unobserved common sources of errors remain after conditioning. Importantly, if researchers want to avoid assuming access to gold-standard labels, they have to make some assumptions about proxies, and our paper makes the assumption explicit and transparent by building and extending the classical literature. 

\begin{remark}
As LLM annotations are one of the most common application areas, we will come back to this assumption and use Section~\ref{sec:conditional-independence-practice} to specifically discuss (a) how to make this assumption more plausible by carefully selecting conditioning variables $\widetilde{D}$ and (b) how to choose different families of LLMs and prompts to construct measurements. We also develop a series of diagnostic tools for this assumption. Please see Section~\ref{sec:conditional-independence-practice}.
\end{remark}

\subsection{Identification}
\label{subsec:iden}
We now discuss how to identify the downstream moment function under the conditional independence assumption. Here we provide an overview of the two steps, which we elaborate in the subsequent sections. 

The first step is to identify the conditional classification rate defined as 
\begin{equation}
    \eta^\ast_{j,a}(d)
    =
    \Pr(X^{(j)}=1\mid X^\ast=a,\widetilde D=d),
\end{equation}
where $\eta^\ast_{j,1}(d)$ denotes the true positive rate and $\eta^\ast_{j,0}(d)$ represents the false positive rate for a specific proxy $j \in \{1, \ldots, J\}$. Here, we build on general nonparametric identification results using array decomposition \citep{kruskal1977three, allman2009identifiability}. 

Second, given the identified conditional classification rate $\eta^\ast_{j,a}(d)$, we construct a \textit{bridge function} that connects the observed proxy $X^{(j)}$ and the unobserved true label $X^\ast$ \citep{zhou2024causal, guo2026proximal}:
\begin{equation*}
    M_j(\widetilde D)
    =
    \frac{X^{(j)}-\eta^\ast_{j,0}(\widetilde D)}
         {\eta^\ast_{j,1}(\widetilde D)-\eta^\ast_{j,0}(\widetilde D)}.
\end{equation*}
This is a tailored function in that, by construction, it is unbiased for the true label, satisfying 
\begin{equation*}
\E[M_j(\widetilde D)\mid X^\ast,\widetilde D] = X^\ast.
\end{equation*}
More generally, as shown below, we can combine multiple proxy labels to construct a general bridge function $H(\widetilde X,\widetilde D)$ that is unbiased for the true latent variable. 
\begin{equation*}
    \E[H(\widetilde X,\widetilde D)
    \mid X^\ast,\widetilde D]
    =
    X^\ast.
\end{equation*}
We can then use this bridge function in place of the true label to build the unbiased, observed downstream moment.
\begin{equation*}
\begin{aligned}
    \psi_{\text{DMM}}(\widetilde X,\widetilde D;\beta)
     \ = \
    \{1-H(\widetilde X,\widetilde D)\}
    \psi^F(Y,0,W;\beta)
    \ + \ 
    H(\widetilde X,\widetilde D)
    \psi^F(Y,1, W;\beta).
\end{aligned}
\end{equation*}
Under Assumption~\ref{ass:label-conditional-independence}, 
$\E[\psi_{\text{DMM}}(\widetilde X,\widetilde D;\beta)] = \E[\psi^F(Y,X^\ast, W;\beta)],$ which will allow for consistent estimation and valid inference as developed in Section~\ref{sec:latent-independent-estimation}.  

\subsubsection{Identification of Conditional Classification Rates}
We start with identification of conditional classification rate $\eta^\ast_{j,a}(d)$. Assumption~\ref{ass:label-conditional-independence} implies that the joint distribution of imperfect proxies can be factorized as follows.
\begin{equation}
    \Pr(\widetilde{X} = \widetilde{x} \mid \widetilde{D} = d)
    =
    \sum_{a=0}^1
    \Pr(X^\ast = a \mid \widetilde{D} = d)
    \prod_{j=1}^J
    \eta^\ast_{j,a}(d)^{x^{(j)}}
    \{1-\eta^\ast_{j,a}(d)\}^{1-x^{(j)}}.
    \label{eq:proxy-mixture}
\end{equation}
This equation mirrors the classical latent-class decomposition and its identification through array decomposition under Kruskal's uniqueness condition \citep{kruskal1977three}. Adapting Theorem~4 of \citet{allman2009identifiability}, we partition $J$ proxies into three mutually disjoint and nonempty subsets $S_1, S_2, S_3$, i.e., $\{1,\ldots,J\} = S_1 \cup S_2 \cup S_3$.
For each
$b\in\{1,2,3\}$, define the corresponding composite measurement $\widetilde X_{S_b} \coloneq \bigl(X^{(j)}\bigr)_{j\in S_b}$,
which takes values in $\{0,1\}^{|S_b|}$ and therefore has $2^{|S_b|}$
possible label patterns. 
Note that the identified form does not depend on how we construct these partitions \citep{allman2009identifiability}.
When $J=3$, this construction reduces to the singleton partition $S_b=\{b\}$ for $b\in\{1,2,3\}$.

For joint label patterns
$x_b\in\{0,1\}^{|S_b|}$, $b\in\{1,2,3\}$, define the conditional
probability tensor
$\mathcal T_d\in
\mathbb R^{2^{|S_1|}\times2^{|S_2|}\times2^{|S_3|}}$ by
\begin{equation*}
    [\mathcal T_d]_{x_1,x_2,x_3}
    \coloneq
    \Pr\!\left(
        \widetilde X_{S_1}=x_1,
        \widetilde X_{S_2}=x_2,
        \widetilde X_{S_3}=x_3
        \mid \widetilde D=d
    \right).
\end{equation*}
For block $S_b$ and latent class $a\in\{0,1\}$, define the block-level
conditional probability matrix:
\begin{equation*}
    [\mathbf S_b(d)]_{x_b,a}
    \coloneq
    \Pr\!\left(
        \widetilde X_{S_b}=x_b
        \mid X^\ast=a,\widetilde D=d
    \right),
    \qquad
    x_b\in\{0,1\}^{|S_b|}.
\end{equation*}
The conditional classification rate is then the corresponding marginals of the
block distribution. In particular, for every $j\in S_b$, $\eta^\ast_{j,a}(d) = \sum_{x_b: x_j=1}[\mathbf S_b(d)]_{x_b,a}$.

Under Assumption~\ref{ass:label-conditional-independence}, the tensor admits
a two-component Candecomp/Parafac (CP) decomposition \citep{carroll1970analysis, harshman1970foundations, stegeman2007kruskal}:
\begin{equation}
    \mathcal T_d
    =
    \sum_{a=0}^1
    \Pr(X^\ast=a\mid\widetilde D=d)
    \mathbf S_{1,\cdot a}(d)
    \otimes
    \mathbf S_{2,\cdot a}(d)
    \otimes
    \mathbf S_{3,\cdot a}(d), \label{eq:cp-decom}
\end{equation}
where
$\otimes$ denotes the outer product of vectors.
Provided that each latent class has positive conditional probability, Kruskal's
uniqueness theorem implies that the CP decomposition in
equation~\eqref{eq:cp-decom} is unique if
\begin{equation}
    k_1(d)+k_2(d)+k_3(d)
    \geq
    2K_\ast+2
    \label{eq:kruskal-rank-condition}
\end{equation}
where $K_\ast$ denotes the number of latent classes, and $k_b(d)\coloneq\operatorname{krank}\{\mathbf S_b(d)\}$
denotes the Kruskal rank of the $b$th block level probability matrix.
Because the columns of each $\mathbf S_b(d)$ are probability vectors,
their scaling is fixed by the requirement that they sum to one.
Therefore, under equation~\eqref{eq:kruskal-rank-condition},
$\{\mathbf S_b(d)\}_{b=1}^{3}$ are
identified up to a common permutation of the latent classes
\citep{kruskal1977three,allman2009identifiability}.

In our binary setting, $K_\ast = 2$.
If each block contains at least one
proxy that is informative about $X^\ast$, then
$\mathbf S_{b,\cdot 0}(d)\neq\mathbf S_{b,\cdot 1}(d)$ and hence
$k_b(d)=2$ for every $b\in\{1,2,3\}$. It follows that
$k_1(d)+k_2(d)+k_3(d) = 2+2+2 = 2K_\ast +2$,
so Kruskal's rank condition holds with equality.
Each proxy-specific
conditional classification rate $\eta^\ast_{j,a}(d)$ can then be identified
by marginalizing the corresponding identified block-level distribution.

An additional anchoring rule resolves the remaining permutation ambiguity.
For example, one may designate an anchor proxy $j^\dagger$ and assume $\eta^\ast_{j^\dagger,1}(\widetilde D) >  \eta^\ast_{j^\dagger,0}(\widetilde D)$ a.s.
This condition states that, for the chosen anchor proxy, the true positive rate is higher than the false positive rate, which is a mild requirement commonly satisfied in modern applications of AI-generated data. More general substantive or
agreement-based anchoring rules are also possible.
At the population level, the identified proxy-specific conditional classification rates do not depend on the choice of tripartition, provided that the same anchoring rule is used to orient the latent classes.

For completeness, we reproduce the following known identification result.
\begin{result}[Identification of the measurement model; \citealp{allman2009identifiability}]
\label{res:measurement-model-identification}
Suppose Assumption~\ref{ass:label-conditional-independence} holds. For
$P_{\widetilde D}$-almost every $d$, suppose that both latent classes
have positive conditional probability, there exists a tripartition into
three nonempty blocks whose block-level probability matrices satisfy
\begin{equation*}
    k_1(d)+k_2(d)+k_3(d)\ge 2K^\ast+2,
\end{equation*}
and a prespecified anchoring rule uniquely orients the two latent classes.
Suppose also that the resulting anchored conditional probabilities admit
measurable versions. Then the conditional mixing probability and the
conditional classification rates
$\{\eta^\ast_{j,a}(d):j=1,\ldots,J,\ a\in\{0,1\}\}$
are identified from the conditional distribution of
$\widetilde X$ given $\widetilde D=d$, up to
$P_{\widetilde D}$-null sets.
\end{result}

\subsubsection{Identification via the Robust Bridge Function}
Given the identified conditional classification rate $\eta^\ast_{j,a}(\widetilde D)$, we now construct a tailored bridge function to obtain an unbiased downstream moment function by building on the recent literature in causal inference \citep{zhou2024causal, guo2026proximal}. We first define a single-proxy bridge as
\begin{equation*}
    M_j(\widetilde D; \eta^\ast)
    =
    \frac{X^{(j)}-\eta^\ast_{j,0}(\widetilde D)}
         {\eta^\ast_{j,1}(\widetilde D) - \eta^\ast_{j,0}(\widetilde D)},
\end{equation*}
which by construction satisfies $\E[M_j(\widetilde D; \eta^\ast)\mid X^\ast,\widetilde D] = X^\ast$.
Assumption~\ref{ass:label-conditional-independence} further implies that, for every nonempty set of distinct labels $S$,
\begin{equation*}
    \E\!\left[
        \prod_{j\in S}M_j(\widetilde D;\eta^\ast)
        \;\middle|\;
        X^\ast,\widetilde D
    \right]
    =
    \prod_{j\in S}
    \E\!\left[
        M_j(\widetilde D;\eta^\ast)
        \mid X^\ast,\widetilde D
    \right]
    =
    (X^\ast)^{|S|}
    =
    X^\ast,
\end{equation*}
where the final equality follows because $X^\ast\in\{0,1\}$.
Thus, many functions of imperfect labels are valid bridges. Consider a class of polynomial bridge functions:
\begin{equation}
    H(\widetilde{X}, \widetilde D; \eta^\ast)
    := 
    \sum_{S \in \mathcal{S}} \omega_{S}
    \prod_{j \in S} M_j(\widetilde D; \eta^\ast) \label{eq:general-bridge}
\end{equation}
where $\mathcal{S}$ is a collection of nonempty subsets of
$\{1,\ldots,J\}$ and $\omega_S$ is the coefficient assigned to the corresponding polynomial term.
Any member of this class preserves the unbiasedness for $X^\ast$ whenever $\sum_{S \in \mathcal{S}} \omega_{S} = 1$.

To obtain desirable statistical properties within this class, we combine multiple proxies in a way that is symmetric between labels and Neyman orthogonal \citep{robins1994, cher2018double} so that it is locally robust to the first-stage estimation error for the conditional classification rate $\eta^\ast_{j, a}(d)$. Define the average pair and triple bridges as
\begin{align*}
    \overline H_2(\widetilde{X}, \widetilde D; \eta^\ast)
    &\coloneq
    \binom{J}{2}^{-1}
    \sum_{j_1<j_2}M_{j_1}(\widetilde D; \eta^\ast)M_{j_2}(\widetilde D; \eta^\ast),
    \\
    \overline H_3(\widetilde{X}, \widetilde D; \eta^\ast)
    &\coloneq
    \binom{J}{3}^{-1}
    \sum_{j_1<j_2<j_3}M_{j_1}(\widetilde D; \eta^\ast)M_{j_2}(\widetilde D; \eta^\ast)M_{j_3}(\widetilde D; \eta^\ast).
\end{align*}
One example of a symmetric robust bridge function is written as follows.
\begin{equation}
    H^{\mathrm{R}}(\widetilde{X}, \widetilde D; \eta^\ast)
    =
    3\overline H_2(\widetilde{X}, \widetilde D; \eta^\ast)-2\overline H_3(\widetilde{X}, \widetilde D; \eta^\ast).
    \label{eq:robust-symmetric-bridge}
\end{equation}
For a symmetric combination
$c_2\overline H_2+c_3\overline H_3$, the unbiasedness of the bridge function requires
$c_2+c_3=1$, while the Neyman orthogonality (i.e., cancellation of the first-order bridge error) under
$X^\ast=1$ requires $2c_2+3c_3=0$, yielding
$c_2=3$ and $c_3=-2$ as a unique set of weights. Under $X^\ast=0$, the pair and triple terms are already at least second order and achieve the Neyman orthogonality. Thus, while a single-proxy bridge is enough for identifying the downstream moment, at least three distinct proxies are needed for the Neyman orthogonality and the resulting robustness to the first-stage estimation error. 

For $J=3$, equation~\eqref{eq:robust-symmetric-bridge} becomes $H^{\mathrm{R}} = M_1M_2+M_1M_3+M_2M_3-2M_1M_2M_3.$ In the binary latent dependent variable setting, this bridge coincides with the construction first proposed by \citet{guo2026proximal}, who also extend this to causal inference with finite-support hidden dependent variables. Theoretically, we generalize the existing results to more than three proxies, examine efficiency gains from additional proxies, and develop a general downstream moment framework covering both latent independent- and dependent-variable settings. This is important as researchers can often construct many proxies from LLMs in the modern applications of AI-generated data. Below, we focus on a robust bridge function with pair and triple products for numerical stability and the sake of clear presentation. However, we emphasize that our method and proofs are applicable to a more general class of bridge functions that use higher-order interaction terms of single-proxy bridge functions (equation~\eqref{eq:general-bridge}). 

We now formally state the unbiasedness and local robustness of the symmetric bridge function. 
\begin{proposition}[Unbiasedness and local robustness of the symmetric robust bridge function]
\label{prop:robust-symmetric-bridge}
We assume that imperfect measurements $\{X^{(j)}\}_{j=1}^J$  are relevant, i.e., for some constant $c_\Delta>0$,
\begin{equation*}
    \min_{1\le j\le J}
    \left|
        \eta^\ast_{j,1}(\widetilde D)-\eta^\ast_{j,0}(\widetilde D)
    \right|
    \ge
    c_\Delta
    \quad\text{a.s.}
\end{equation*}
Under Assumption~\ref{ass:label-conditional-independence}, 
\begin{equation}
    \E\!
    \left[
        H^{\mathrm{R}}(\widetilde{X}, \widetilde D; \eta^\ast)
        \mid X^\ast, \widetilde D
    \right]
    =
    X^\ast.
    \label{eq:robust-symmetric-recovery}
\end{equation}
Moreover, consider cross-fitted estimates of the single-label bridges $\widehat M_j$ and let the error term
\begin{equation*}
    e_{j,a}(\widetilde D_i)
    =
    \E\!\left[
        \widehat M_{ij}(\widetilde D_i; \widehat{\eta}^{(-k(i))})
        \mid X^\ast_i=a,\widetilde D_i,\mathcal I_{-k(i)}
    \right]
    -a,
    \qquad \mbox{for  } a\in\{0,1\},
\end{equation*}
where $\mathcal{I}_{-k(i)}$ is the data excluding the fold $k(i)$ that unit $i$ belongs to, and $\widehat M_{ij}$ is constructed by the conditional classification rate $\widehat{\eta}^{(-k(i))}$ estimated only using $\mathcal{I}_{-k(i)}$. Conditional on the training sample $\mathcal{I}_{-k}$, the fitted robust bridge only contains the second- or third-order bias. 
\begin{align*}
    &\E[\widehat H^{\mathrm{R}}(\widetilde{X}, \widetilde D; \widehat{\eta}^{(-k)})\mid X^\ast=1,\widetilde D,\mathcal I_{-k}]-1
    \nonumber\\
    &\qquad=
    -3\binom{J}{2}^{-1}
    \sum_{j_1<j_2}e_{j_1,1}(\widetilde D)e_{j_2,1}(\widetilde D)
    -2\binom{J}{3}^{-1}
    \sum_{j_1<j_2<j_3}e_{j_1,1}(\widetilde D)e_{j_2,1}(\widetilde D)e_{j_3,1}(\widetilde D),
    \\
    &\E[\widehat H^{\mathrm{R}}(\widetilde{X}, \widetilde D; \widehat{\eta}^{(-k)})\mid X^\ast=0,\widetilde D,\mathcal I_{-k}]
    \nonumber\\
    &\qquad=
    3\binom{J}{2}^{-1}
    \sum_{j_1<j_2}e_{j_1,0}(\widetilde D)e_{j_2,0}(\widetilde D)
    -2\binom{J}{3}^{-1}
    \sum_{j_1<j_2<j_3}e_{j_1,0}(\widetilde D)e_{j_2,0}(\widetilde D)e_{j_3,0}(\widetilde D).
\end{align*}
Consequently, the conditional bridge error contains no term that is linear in
a single first-stage error. 
\end{proposition}
The proof is provided in Appendix~\ref{app:proof-rs-bridge}.
As shown above, the coefficients $3$ and $-2$ preserve the unbiasedness of the bridge function and eliminate first-order bridge error: at $X^\ast=1$, $3(2/J)-2(3/J)=0$, while at $X^\ast=0$, no linear term arises because singleton products are excluded. 
Therefore, if each estimated conditional classification rate model converges in $L_2(P)$ at rate $\delta_n$ and the downstream moment contrast is uniformly bounded as assumed in Theorem~\ref{thm:latent-independent-asymptotics}, the resulting population remainder of the bridge-based moment function is of order
$O_p(\delta_n^2)$, rather than $O_p(\delta_n)$. This second-order remainder
underlies the asymptotic inference results in the next subsection.

Finally, we can now substitute the robust bridge into the binary decomposition of the full-data moment equation to obtain an unbiased downstream moment function.
\begin{equation*}
\begin{aligned}
    \psi^{\mathrm{DMM}}( \widetilde{X}, \widetilde{D};\beta, \eta^\ast)
    \ = \ 
    \{1-H^{\mathrm{R}}(\widetilde{X}, \widetilde D; \eta^\ast)\}
    \psi^F(Y, 0, W;\beta)
    \ + \ 
    H^{\mathrm{R}}(\widetilde{X}, \widetilde D; \eta^\ast)
    \psi^F(Y, 1, W;\beta).
\end{aligned}
\end{equation*}

\begin{theorem}[Identification of the downstream parameter]
\label{thm:latent-independent-identification}
Suppose the conditions of
Result~\ref{res:measurement-model-identification} hold, and suppose
that the class contrasts of all proxies used in the robust bridge are
bounded away from zero. Suppose also that the oracle downstream moment
in equation~\eqref{eq:full-data-target} has the unique solution $\beta^\ast$. Then the robust
bridge is identified from the observed-data law and, for every $\beta$,
\begin{equation*}
    \E[\psi^{\mathrm{DMM}}( \widetilde{X}, \widetilde{D};\beta, \eta^\ast)]
    =
    \E[\psi^F(Y, X^\ast,W;\beta)]
\end{equation*}
where $\eta^\ast(\cdot)$ denotes the true conditional classification rate function. 
Consequently, $\beta^\ast$ is identified from the observed-data law as
the unique solution to 
\begin{equation*}
    \E[\psi^{\mathrm{DMM}}( \widetilde{X}, \widetilde{D};\beta, \eta^\ast)]=0.
\end{equation*}
\end{theorem}

\begin{proof}
By Result~\ref{res:measurement-model-identification},
$\eta^\ast$ is identified from the observed-data law. The class-contrast
condition therefore makes the robust bridge $H^R$ well defined and
identified.
By iterated expectation and equation~\eqref{eq:robust-symmetric-recovery},
\begin{align*}
    \E[\psi^{\mathrm{DMM}}( \widetilde{X}, \widetilde{D};\beta, \eta^\ast)]
    & =
    \E\left[
        \psi^F(Y, 0, W;\beta)
        +
        \E[H^{\mathrm{R}}(\widetilde{X}, \widetilde D; \eta^\ast)\mid  X^\ast,\widetilde D]
        \{\psi^F(Y, 1, W;\beta) - \psi^F(Y, 0, W;\beta)\}
    \right]
    \\
    & =
    \E\left[
        (1-X^\ast)\psi^F(Y, 0, W;\beta)
        +X^\ast\psi^F(Y, 1, W;\beta)
    \right]\\
    & =
    \E\left[\psi^F(Y, X^\ast, W;\beta)\right],
\end{align*}
which completes the proof.
\end{proof}

\subsection{Estimation and Inference}
\label{sec:latent-independent-estimation}

Given the identification results in the previous section, this section proposes the DMM (debiased inference with multiple imperfect measurements) estimator and proves its consistency and asymptotic normality.

First, to estimate the conditional classification rate, we use cross-fitting \citep{cher2018double} and separate estimation of the nuisance function from
evaluation of the downstream moment function. Partition the observations into
$K$ folds $\mathcal I_1,\ldots,\mathcal I_K$. For each fold $k$, estimate the
conditional classification rate model $\{\eta_{j,0}(d), \eta_{j,1}(d)\}_{j=1}^J$ from $(\widetilde X_i,\widetilde D_i)$ using
observations outside $\mathcal I_k$. One option is to maximize the conditional
likelihood
\begin{equation*}
\begin{aligned}
    \mathcal L_{-k}
    =
    \sum_{i\notin\mathcal I_k}
    \log\Bigg[
        \sum_{a=0}^1
        \pi(\widetilde D_i)^a
        \{1-\pi(\widetilde D_i)\}^{1-a}
        \prod_{j=1}^J
        \eta_{j,a}(\widetilde D_i)^{X_i^{(j)}}
        \{1-\eta_{j,a}(\widetilde D_i)\}^{1-X_i^{(j)}}
    \Bigg],
\end{aligned}
\end{equation*}
where $\pi(\widetilde{D})\coloneq\Pr(X^\ast = 1 \mid \widetilde D).$
This likelihood can be maximized using the expectation-maximization (EM) algorithm \citep{dempster1977maximum}, as in \citet{guo2026proximal}.
More flexible nuisance functions can be accommodated
by maximizing the same likelihood over sieve spaces \citep{shen1997methods, hu2008instrumental}, 
or alternatively by sieve minimum distance based on conditional moment restrictions \citep{chen2012estimation, zhou2024causal}.
Each fold-specific fit is oriented using the same anchoring rule. As in debiased machine learning \citep{cher2018double}, the asymptotic results below do not depend on which particular fitting algorithm researchers use to estimate the conditional classification rates. Below, we provide high-level requirements for convergence rates that can be achieved by a wide range of fitting algorithms. 

Then, for $i\in\mathcal I_k$, we construct
\begin{equation}
    \widehat M_{ij}^{(-k)}
    =
    \frac{X_i^{(j)}-\widehat\eta_{j,0}^{(-k(i))}(\widetilde D_i)}{\widehat\eta_{j,1}^{(-k(i))}(\widetilde D_i)-\widehat\eta_{j,0}^{(-k(i))}(\widetilde D_i)}, \label{eq:est-bridge}
\end{equation}
where $\widehat{\eta}^{(-k(i))}$ is estimated only using $\mathcal{I}_{-k(i)}$, the data excluding the fold $k(i)$ that unit $i$ belongs to. Then, the cross-fitted moment function for the DMM estimator is 
\begin{equation}
\begin{aligned}
    \widehat\psi^{\mathrm{DMM}}(\widetilde{X}_i, \widetilde{D}_i; \beta, \widehat{\eta}^{(-k(i))})
    \ \coloneq \ 
    \{1-\widehat H_i^{\mathrm{R}}(\widehat{\eta}^{(-k(i))})\}
    \psi^F(Y_i,0, W_i;\beta)
    \ + \ 
    \widehat H_i^{\mathrm{R}}(\widehat{\eta}^{(-k(i))})
    \psi^F(Y_i, 1, W_i;\beta),
\end{aligned}
\label{eq:estimated-robust-symmetric-score}
\end{equation}
where $\widehat{H}_i^{\mathrm{R}}(\widehat{\eta}^{(-k(i))})$ is defined by equation~\eqref{eq:robust-symmetric-bridge} with estimated single-proxy bridge functions in equation~\eqref{eq:est-bridge}. Our DMM estimator $\widehat\beta^{\mathrm{DMM}}$ can then be written as the solution to the following equation.
\begin{equation*}
    \frac{1}{n}
    \sum_{i=1}^n
    \widehat\psi^{\mathrm{DMM}}(\widetilde{X}_i, \widetilde{D}_i; \beta, \widehat{\eta}^{(-k(i))})
    =0.
\end{equation*}
Throughout this subsection, the observations are i.i.d., the numbers of
proxies $J$ and folds $K$ are fixed, each fold contains a nonvanishing
fraction of the sample, and all fold-specific nuisance estimates are
oriented using the same anchoring rule.

The following theorem presents the asymptotic properties of the DMM estimator. 

\begin{theorem}[Consistency and asymptotic normality]
\label{thm:latent-independent-asymptotics}
Suppose the assumptions required in
Theorem~\ref{thm:latent-independent-identification} hold. Let
$\mathcal B_0\subset\mathcal B$ be a compact neighborhood of
$\beta^\ast$, and let $\widehat\beta^{\mathrm{DMM}}$ be a solution
of the sample moment equation in $\mathcal B_0$, which exists
with probability approaching one. We assume the following standard regularity conditions:
(a) $\beta^\ast$ is the unique well-separated interior root on $\mathcal B_0$;
(b) for $a\in\{0,1\}$, the moment functions
$\psi^F(Y,a,W;\beta)$ and their Jacobians are Lipschitz in $\beta$
on $\mathcal B_0$ with square-integrable envelopes, 
the population
Jacobian $\mathbf A_0$ defined below is nonsingular,
and there exists
a finite constant $C$ such that
$\sup_{\beta\in\mathcal B_0}\left\|\psi^F(Y,1,W;\beta)-\psi^F(Y,0,W;\beta)\right\|\le C$ a.s.; 
(c) 
the true and fitted class contrasts, $\{\eta^\ast_{j,1}(\widetilde D) - \eta^\ast_{j,0}(\widetilde D)\}_{j=1}^J$ and $\{\{\widehat{\eta}^{(-k)}_{j,1}(\widetilde D) - \widehat{\eta}^{(-k)}_{j,0}(\widetilde D)\}_{j=1}^J\}_{k=1}^K$, are uniformly bounded away
from zero over $\widetilde D$, the latter with probability
approaching one, and (d) for every $k$, $j$, $a$, $\widehat{\eta}^{(-k)}_{j,a}(\widetilde D)$ takes values in $[0, 1]$ almost surely.

Define the maximum fold-specific nuisance error.
\begin{equation*}
    \delta_n
    =
    \max_{k,j,a}
    \left\|
        \widehat\eta_{j,a}^{(-k)}-\eta_{j,a}^\ast
    \right\|_{2,P},
\end{equation*}
where $\eta_{j,a}^\ast$ is the true classification rate and the norm is taken over the distribution of $\widetilde D$. If the nuisance function estimation is consistent, i.e., $\delta_n=o_p(1)$, then
$\widehat\beta^{\mathrm{DMM}}\overset{p}{\to}\beta^\ast$. More generally, the DMM estimator achieves multiple robustness for consistency: consistency of the DMM estimator requires the convergence of the nuisance functions to the true functions for only $J-1$ proxies, and one does not need to know which $J-1$ classification rates are consistently estimated.

If the nuisance function estimator converges at a slow nonparametric rate, i.e., $\delta_n=o_p(n^{-1/4}),$ 
\begin{equation}
    \sqrt n
    (\widehat\beta^{\mathrm{DMM}}-\beta^\ast)
    =
    \frac{1}{\sqrt n}
    \sum_{i=1}^n
    \mathbf{A}_0^{-1}
    \psi^{\mathrm{DMM}}( \widetilde{X}_i, \widetilde{D}_i; \beta^\ast, \eta^\ast)
    +o_p(1),
    \label{eq:rs-asymptotic-linearity}
\end{equation}
and
\begin{equation}
    \sqrt n
    (\widehat\beta^{\mathrm{DMM}}-\beta^\ast)
    \rightsquigarrow
    N\!\left(
        0,
        \mathbf{A}_0^{-1}\boldsymbol{\Omega}_0\mathbf{A}_0^{-\top}
    \right)
    \label{eq:rs-asymptotic-normality}
\end{equation}
where 
\begin{align*}
    \mathbf{A}_0
    &=
    -\E\!\left[
        \{1-H^{\mathrm{R}}(\widetilde X,\widetilde D; \eta^\ast)\} \frac{\partial}{\partial \beta^\top} \psi^F(Y,0, W; \beta^\ast)        
        +
        H^{\mathrm{R}}(\widetilde X,\widetilde D; \eta^\ast) \frac{\partial}{\partial \beta^\top} \psi^F(Y,1, W; \beta^\ast)
    \right],
    \\
    \boldsymbol{\Omega}_0
    &=
    \E\!\left[
        \psi^{\mathrm{DMM}}( \widetilde{X}, \widetilde{D}; \beta^\ast, \eta^\ast)
        \psi^{\mathrm{DMM}}( \widetilde{X}, \widetilde{D}; \beta^\ast, \eta^\ast)^\top
    \right].
\end{align*}

\end{theorem}

The rate requirement for the nuisance function estimation $\widehat{\eta}$ follows from Proposition~\ref{prop:robust-symmetric-bridge}:
the population remainder caused by estimating the conditional classification rate is second order
in the nuisance function estimation error. Cross-fitting controls the corresponding empirical
process term without requiring restrictive Donsker conditions, following the
logic of locally robust estimation \citep{chernozhukov2022locally}.
A proof is provided in Appendix~\ref{app:proof-rs-asymptotics}.

To estimate the asymptotic variance, we evaluate the cross-fitted moment function at
$\widehat\beta^{\mathrm{DMM}}$ and define
\begin{align}
    \widehat{\mathbf{A}}
    ={}&
    -\frac{1}{n}\sum_{i=1}^n
    \left[
        \{1-\widehat H_i^{\mathrm{R}}(\widehat{\eta}^{(-k(i))})\}
        \frac{\partial}{\partial \beta^\top} \psi^F(Y_i,0, W_i; \widehat{\beta}^{\mathrm{DMM}}) 
        +
        \widehat H_i^{\mathrm{R}}(\widehat{\eta}^{(-k(i))})
        \frac{\partial}{\partial \beta^\top} \psi^F(Y_i,1, W_i; \widehat{\beta}^{\mathrm{DMM}})
    \right],
    \label{eq:estimated-rs-bread}
    \\
    \widehat{\boldsymbol{\Omega}}
    ={}&
    \frac{1}{n}\sum_{i=1}^n
    \psi^{\mathrm{DMM}}( \widetilde{X}_i, \widetilde{D}_i; \widehat{\beta}^{\mathrm{DMM}}, \widehat{\eta}^{(-k(i))})
        \psi^{\mathrm{DMM}}( \widetilde{X}_i, \widetilde{D}_i; \widehat{\beta}^{\mathrm{DMM}}, \widehat{\eta}^{(-k(i))})^\top,
    \\
    \widehat{\mathbf{V}}
    ={}&
    \widehat{\mathbf{A}}^{-1}
    \widehat{\boldsymbol{\Omega}}
    \widehat{\mathbf{A}}^{-\top}.
    \label{eq:estimated-rs-variance}
\end{align}
Then $\widehat{\mathbf{V}}/n$ estimates the covariance matrix of
$\widehat\beta^{\mathrm{DMM}}$. For any fixed contrast $c$, an asymptotic
$(1-\alpha)$ confidence interval for $c^\top\widehat\beta^{\mathrm{DMM}}$ is
\begin{equation*}
    c^\top\widehat\beta^{\mathrm{DMM}}
    \;\pm\;
    z_{1-\alpha/2}
    \sqrt{
        \frac{c^\top\widehat{\mathbf{V}}c}{n}
    },
\end{equation*}
where $z_{1-\alpha/2}$ is the $(1-\alpha/2)$ percentile of the standard normal distribution. 
For the generalized linear model example in Section~\ref{sec:problem-setting},
$c$ may select the coefficient on the latent regressor.

\subsubsection*{Discussion on Efficiency.}
The robust bridge is designed to remove first-order nuisance bias.
Here, we examine how the asymptotic precision of this estimator
changes with the number and quality of the labels. 

Write $H_{0,J}^{\mathrm{R}}$ to emphasize that the bridge uses $J$ proxies. Then, the DMM moment function admits the decomposition
\begin{equation}
    \psi_{J}^{\mathrm{DMM}}( \widetilde{X}, \widetilde{D}; \beta^\ast, \eta^\ast)
    =
    \psi^F(Y,X^\ast,W; \beta^\ast)
    +\{H_{0,J}^{\mathrm{R}}-X^\ast\}
    \Delta\psi_0(Y,W),
    \label{eq:variance-rs-score-decomposition}
\end{equation}
where
$\Delta\psi_0(Y,W) = \psi^F(Y, 1, W;\beta^\ast)-\psi^F(Y,0, W;\beta^\ast)$. The unbiasedness of the bridge function implies that
the population bread matrix equals its oracle counterpart.
\begin{equation}
    \mathbf A_0
    =
   -\E\left[
        \left.
        \frac{\partial\psi^F(Y,X^\ast,W;\beta)}
        {\partial\beta^\top}
        \right|_{\beta=\beta^\ast}
    \right].
    \label{eq:variance-oracle-bread}
\end{equation}
Thus, the number and quality of the labels
affect the asymptotic variance only through the meat matrix:
\begin{equation}
    \boldsymbol\Omega_{0,J}^{\mathrm{R}}
    =
    \boldsymbol\Omega_0^F
    +
    \E\left[
        \{H_{0,J}^{\mathrm{R}}-X^\ast\}^2
        \Delta\psi_0(Y,W)\Delta\psi_0(Y,W)^\top
    \right]
    \label{eq:variance-meat-decomposition}
\end{equation}
where $\boldsymbol\Omega_0^F \coloneq \E\left[\psi^F(Y,X^\ast, W;\beta^\ast)\psi^{F}(Y,X^\ast, W;\beta^\ast)^\top\right]$. We use $\mathbf V_{0,J}^{\mathrm{R}} \coloneq \mathbf A_0^{-1} \boldsymbol\Omega_{0,J}^{\mathrm{R}} \mathbf A_0^{-\top}$ to denote the asymptotic variance of the DMM estimator using $J$ proxies, and we use  
$\mathbf V_0^F
    \coloneq
    \mathbf A_0^{-1}
    \boldsymbol\Omega_0^F
    \mathbf A_0^{-\top}
$ to denote the asymptotic variance of an infeasible, oracle estimator that uses unobserved $X^\ast$ in downstream moments.  

For a fixed linear combination of coefficients $c^\top\beta^\ast$, we can write the asymptotic variance of the DMM estimator as follows. 
\begin{equation}
\begin{aligned}
    c^\top\mathbf V_{0,J}^{\mathrm{R}}c
     \ = \ 
    c^\top\mathbf V_0^F c
    \ + \ 
    \E\left[
        \sum_{a=0}^1
        \pi(\widetilde D_i)^a
        \{1-\pi(\widetilde D_i)\}^{1-a}
        v_{a,J}^{\mathrm{R}}(\widetilde D)
        \{(\mathbf{A}_0^{-\top}c)^\top\Delta\psi_0(Y,W)\}^2
    \right],
\end{aligned}
\label{eq:variance-coefficient-factorization}
\end{equation}
where $v_{a,J}^{\mathrm{R}}(d) \coloneq \V\left(H_{0,J}^{\mathrm{R}}\mid X^\ast=a,\widetilde D=d\right)$.
Importantly, only $v_{a,J}^{\mathrm{R}}(d)$ depends on the selection of proxies. Hence, decreasing $v_{a,J}^{\mathrm{R}}(d)$ is sufficient to improve precision for downstream parameters of interest.

For proxy $j$, define its conditional single bridge variance
\begin{equation*}
\kappa_{j,a}(d) \coloneq \frac{\eta^\ast_{j,a}(d)\{1-\eta^\ast_{j,a}(d)\}}{\{\eta^\ast_{j,1}(d) - \eta^\ast_{j,0}(d)\}^2}.
\end{equation*}
Then, we can express the conditional variance of the bridge function as follows.
\begin{equation}
    v_{a,J}^{\mathrm{R}}(d)
    \ =  \ 
    \frac{9}{\binom{J}{2}^{2}}
    \sum_{j_1<j_2}
    \kappa_{j_1,a}(d)\kappa_{j_2,a}(d)
    \ + \
    \frac{4}{\binom{J}{3}^{2}}
    \sum_{j_1<j_2<j_3}
    \kappa_{j_1,a}(d)
    \kappa_{j_2,a}(d)
    \kappa_{j_3,a}(d).
\label{eq:variance-rs-exact}
\end{equation}
Suppose an additional proxy has conditional bridge variance
$\kappa_{J+1,a}(d)$. Appendix~\ref{app:rs-efficiency} derives an explicit
threshold $\overline{\kappa}_{J,a}(d)$ such that, except in the degenerate zero-variance
case,
\begin{equation*}
    v_{a,J+1}^{\mathrm{R}}(d)
    <
    v_{a,J}^{\mathrm{R}}(d)
    \quad\Longleftrightarrow\quad
    \kappa_{J+1,a}(d)< \overline{\kappa}_{J,a}(d).
\end{equation*}
Thus, adding a proxy improves precision when the conditional variance of its single-proxy bridge function is
small enough relative to those of the existing proxies. For example, a proxy with a very
small class contrast, $\eta_{j,1}(d) - \eta_{j,0}(d)$, has a large $\kappa_{j,a}(d)$ and can increase the variance of the bridge function $v_{a,J+1}^{\mathrm{R}}(d)$. More proxies therefore do not mechanically improve this particular estimator. It is straightforward to allow for different weights when constructing the robust bridge function so that it can guarantee the DMM estimator with $J+1$ proxies is at least as efficient as the DMM estimator with $J$ proxies. 

In a special case where every proxy has the same conditional variance $\kappa_a(d)$,
\begin{equation*}
    v_{a,J}^{\mathrm{R}}(d)
    =
    \frac{18\kappa_a(d)^2}{J(J-1)}
    +
    \frac{24\kappa_a(d)^3}{J(J-1)(J-2)},
\end{equation*}
which decreases with $J$ at rate $J^{-2}$, i.e., more proxies lead to higher accuracy for downstream parameter estimation.

\section{Extensions}
\label{sec:latent-dependent}

We now extend the proposed framework to settings in which an error-prone variable is a binary \textit{dependent variable} (Section~\ref{subsec:latent-dependent}). We then provide the most general results where we extend our method to settings where either an independent or dependent variable is \textit{multicategory} (Section~\ref{subsec:multicategory-labels}).

\subsection{Latent Dependent Variable}
\label{subsec:latent-dependent}

\subsubsection{Setup and Identification}
Let $Y_i^\ast\in\{0,1\}$ denote the latent true outcome, and let $W_i$ collect all observed variables entering the downstream analysis as before, including any focal independent variable and additional covariates. 
Instead of $Y_i^\ast$, the researcher observes $J\geq3$ imperfect labels
\begin{equation*}
    \widetilde Y_i
    =
    \bigl(Y_i^{(1)},\ldots,Y_i^{(J)}\bigr),
    \qquad
    Y_i^{(j)}\in\{0,1\}.
\end{equation*}
As in Section~\ref{sec:latent-independent}, let $D_i$ collect auxiliary information that may help account for heterogeneity or
shared dependence in the label errors. Define $\widetilde D_i = (D_i,W_i)$.  Let $\phi^F(y,w;\beta)\in\mathbb R^{d_\beta}$ be a user-specified full-data moment function. Reusing $\beta^\ast$ for the target parameter in this
section, suppose that if $Y^\ast$ were observed, $\beta^\ast$ would be the unique
solution to
\begin{equation*}
    \E\!\left[
        \phi^F(Y^\ast,W;\beta)
    \right]
    =0.
\end{equation*}

We impose an analog of Assumption~\ref{ass:label-conditional-independence} for the latent outcome setting.

\begin{assumption}[Independence across outcome labels conditioning on input and downstream information]
\label{ass:latent-dependent-label-independence}
The imperfect outcome labels are mutually independent conditional on the
latent outcome and the measurement conditioning covariates:
\begin{equation*}
    Y^{(1)}\indep\cdots\indep Y^{(J)}
    \mid Y^\ast,\widetilde D.
\end{equation*}
\end{assumption}
For $a\in\{0,1\}$, let
\begin{equation*}
    \pi^\ast(d)
    \coloneq
    \Pr(Y^\ast=1\mid\widetilde D=d),
    \qquad
    \eta^\ast_{j,a}(d)
    \coloneq
    \Pr(Y^{(j)}=1\mid Y^\ast=a,\widetilde D=d).
\end{equation*}
We impose the analogs of the overlap, relevance, and anchoring conditions in
Section~\ref{sec:latent-independent}. Under these conditions, the
joint distribution of $(\widetilde Y,\widetilde D)$ identifies
$\pi^\ast(d)$ and $\eta^\ast_{j,a}(d)$ by the same conditional latent-class
argument, so we do not repeat the decomposition here.

Define the outcome-proxy bridge
\begin{equation*}
    M_j^Y(\widetilde D)
    =
    \frac{
        Y^{(j)}-\eta^\ast_{j,0}(\widetilde D)
    }{
        \eta^\ast_{j,1}(\widetilde D)
        -\eta^\ast_{j,0}(\widetilde D)
    }.
\end{equation*}
Substituting this outcome-proxy bridge $M_j^Y$ into equation~\eqref{eq:robust-symmetric-bridge} defines the robust
outcome bridge $H_Y^{\mathrm{R}}(\widetilde Y,\widetilde D)$. The same
argument as in Proposition~\ref{prop:robust-symmetric-bridge} gives
\begin{equation}
    \E\!\left[
        H_Y^{\mathrm{R}}(\widetilde Y,\widetilde D)
        \mid Y^\ast,\widetilde D
    \right]
    =
    Y^\ast.
    \label{eq:latent-dependent-bridge-recovery}
\end{equation}
The fitted bridge also retains the same local robustness: its conditional
error contains no term that is linear in a single proxy-specific nuisance
error. We can then define the observed-data moment function
\begin{equation}
\begin{aligned}
    \phi_Y^{\mathrm{DMM}}(\widetilde{Y}, \widetilde{D};\beta,\eta^\ast)
     \ = \ 
    \{1-H_Y^{\mathrm{R}}(\widetilde Y,\widetilde D)\}
    \phi^F(0, W;\beta)
    \ + \ 
    H_Y^{\mathrm{R}}(\widetilde Y,\widetilde D)
    \phi^F(1, W;\beta).
\end{aligned}
\label{eq:latent-dependent-observed-estimating-function}
\end{equation}

\begin{proposition}[Identification with a latent dependent variable]
\label{prop:latent-dependent-identification}
Suppose Assumption~\ref{ass:latent-dependent-label-independence} and the
analogs of the overlap, relevance, and anchoring conditions in
Section~\ref{sec:latent-independent-assumptions} hold. Then, for every
$\beta$,
\begin{equation}
    \E\!\left[
        \phi_Y^{\mathrm{DMM}}(\widetilde{Y}, \widetilde{D};\beta,\eta^\ast)
    \right]
    =
    \E\!\left[
        \phi^F(Y^\ast,W;\beta)
    \right].
    \label{eq:latent-dependent-moment-equivalence}
\end{equation}
Consequently, $\beta^\ast$ is identified as the unique solution to
$\E[\phi_Y^{\mathrm{DMM}}(\widetilde{Y}, \widetilde{D};\beta,\eta^\ast)]=0$.
\end{proposition}

\begin{proof}
Equation~\eqref{eq:latent-dependent-moment-equivalence} follows directly from
equations~\eqref{eq:latent-dependent-bridge-recovery} and~\eqref{eq:latent-dependent-observed-estimating-function} by iterated expectation.
\end{proof}

\subsubsection{Estimation and Inference}

Estimation follows Section~\ref{sec:latent-independent}. Within
each training fold, estimate the conditional latent-class model of
$\widetilde Y$ given $\widetilde D$, apply the same anchoring rule, and
construct the out-of-fold bridge $\widehat H_{Y,i}^{\mathrm{R}}$. Define
\begin{equation*}
    \widehat\phi_{Y,i}^{\mathrm{DMM}}(\beta)
    =
    (1-\widehat H_{Y,i}^{\mathrm{R}})
    \phi^F(0, W_i;\beta)
    +
    \widehat H_{Y,i}^{\mathrm{R}}
    \phi^F(1, W_i;\beta),
\end{equation*}
and let $\widehat\beta_{Y}^{\mathrm{DMM}}$ solve
\begin{equation*}
    \frac{1}{n}
    \sum_{i=1}^n
    \widehat\phi_{Y,i}^{\mathrm{DMM}}(\beta)
    =0.
\end{equation*}

\begin{theorem}[Large-sample theory for a latent dependent variable]
\label{thm:latent-dependent-asymptotics}
Suppose the direct analogs of the conditions in
Theorem~\ref{thm:latent-independent-asymptotics} hold. In particular, for the
same compact neighborhood $\mathcal B_0$, suppose that there exists a finite
constant $C$ such that
$\sup_{\beta\in\mathcal B_0}\left\|\phi^F(1,W;\beta)-\phi^F(0,W;\beta)\right\|\le C$ a.s.
Let
\begin{equation*}
    \delta_{Y,n}
    =
    \max_{k,j,a}
    \left\|
        \widehat\eta_{j,a}^{Y,(-k)}-\eta^\ast_{j,a}
    \right\|_{2,P}.
\end{equation*}

If $\delta_{Y,n}=o_p(1)$, then
$\widehat\beta_{Y}^{\mathrm{DMM}} \overset{p}{\to}\beta^\ast$. If, in addition, $\delta_{Y,n}=o_p(n^{-1/4})$, then
\begin{equation*}
    \sqrt n
    (\widehat\beta_{Y}^{\mathrm{DMM}}-\beta^\ast)
    \rightsquigarrow
    N\!\left(
        0,
        \mathbf{A}_{Y,0}^{-1}
        \boldsymbol{\Omega}_{Y,0}
        \mathbf{A}_{Y,0}^{-\top}
    \right),
\end{equation*}
where
\begin{align*}
    \mathbf{A}_{Y,0}
    &={}
    -\E\!\left[
        (1-H_{Y}^{\mathrm{R}}(\widetilde Y,\widetilde D; \eta^\ast))
        \left.
        \frac{\partial\phi^F(Y^\ast=0,W;\beta)}
        {\partial\beta^\top}
        \right|_{\beta=\beta^\ast}        +
        H_{Y}^{\mathrm{R}}(\widetilde Y,\widetilde D; \eta^\ast)
        \left.
        \frac{\partial\phi^F(Y^\ast=1,W;\beta)}
        {\partial\beta^\top}
        \right|_{\beta=\beta^\ast}    
        \right],
    \\
    \boldsymbol{\Omega}_{Y,0}
    &={}
    \E\!\left[
        \phi_Y^{\mathrm{DMM}}(\widetilde{Y}, \widetilde{D};\beta^\ast,\eta^\ast)
        \phi_Y^{\mathrm{DMM}}(\widetilde{Y}, \widetilde{D};\beta^\ast,\eta^\ast)^\top
    \right].
\end{align*}
\end{theorem}

The result follows from Proposition~\ref{prop:robust-symmetric-bridge} and
Theorem~\ref{thm:latent-independent-asymptotics} after replacing
$(X^\ast,\widetilde X,\psi_a)$ with
$(Y^\ast,\widetilde Y,\phi_a)$. The covariance matrix can be estimated by the
plug-in estimator for $\mathbf{A}_{Y,0}$ and $\boldsymbol{\Omega}_{Y,0}$.

\subsection{Multicategory Labels}
\label{subsec:multicategory-labels}
So far we have focused on settings with a binary independent or dependent variable. In this subsection, we discuss how our method can be generalized to the multicategory case. 
The downstream moment argument extends directly; the main additional challenge lies in identifying a vector of class-specific bridges under stronger rank and anchoring conditions. For the sake of clear presentation, here we present our method for a multicategory dependent variable, but an analogous extension applies to the multicategory latent independent variable setting. As discussed earlier, for causal inference with three proxies, the bridge construction builds on a general result developed in \citet{guo2026proximal} for the finite-support latent outcome variable setting. Here, we connect that result to our general downstream moment framework and develop extensions to more than three proxies and to latent independent variable settings.

Let $Y^\ast\in\{1,\ldots,K_\ast\}$ denote a multicategory latent dependent variable, and suppose that researchers
observe $\widetilde Y = \bigl(Y^{(1)},\ldots,Y^{(J)}\bigr)$, $Y^{(j)}\in\{1,\ldots,K_j\}$.
The labels need not have the same number of categories as one another or as $Y^\ast$.
This may occur, for example, when some human or machine annotators do not use all available categories for annotation.
As before, let
$\widetilde D=(D,W)$ collect the downstream covariates and any auxiliary information used to model the label distribution.

Define $\phi_a(W;\beta) \coloneq \phi^F(a,W;\beta)$ to denote the full-data moment function evaluated at $Y^\ast = a$, now for $a\in\{1,\ldots,K_\ast\}$.
Thus, it admits the exact decomposition:
\begin{equation*}
    \phi^F(Y^\ast,W;\beta)
    =
    \sum_{a=1}^{K_\ast}
    \bbone\{Y^\ast=a\}\phi_a(W;\beta).
\end{equation*}
Therefore, suppose that class-specific bridges
$H_a(\widetilde Y,\widetilde D)$ satisfy
\begin{equation*}
    \E\!\left[
        H_a(\widetilde Y,\widetilde D)
        \mid Y^\ast=b,\widetilde D
    \right]
    =
    \bbone\{a=b\},
    \qquad
    a,b\in\{1,\ldots,K_\ast\}.
\end{equation*}
Then the observed-data moment function
\begin{equation}
    \phi^H(\widetilde Y,\widetilde D;\beta)
    =
    \sum_{a=1}^{K_\ast}
    H_a(\widetilde Y,\widetilde D)
    \phi_a(W;\beta)
    \label{eq:multicategory-observed-moment}
\end{equation}
is unbiased for the oracle downstream full-data moment function. 

To construct the class-specific bridges, let $\mathcal Y = \prod_{j=1}^J\{1,\ldots,K_j\}$ denote the support of the joint proxy vector.
At each value
$\widetilde D=d$, define the joint measurement matrix
\begin{equation*}
    \mathbf{S}(d) 
    =
    \left[
        \Pr(\widetilde Y=\widetilde y
        \mid Y^\ast=a,\widetilde D=d)
    \right]_{\widetilde y\in\mathcal Y,\,
             a\in\{1,\ldots,K_\ast\}}.
\end{equation*}
Under the multicategory analog of
Assumption~\ref{ass:latent-dependent-label-independence}, the observed proxy
distribution is a finite mixture of product distributions. Tensor decomposition identifies the class-specific label distributions, and hence $\mathbf{S}(d)$, up to a common permutation of the latent classes under the
corresponding rank and anchoring conditions
\citep{kruskal1977three,allman2009identifiability}.

Suppose that $\mathbf{S}(d)$ has full column rank, and let $\mathbf{S}^\dagger(d)\coloneq \{\mathbf{S}(d)^\top \mathbf{S}(d)\}^{-1}\mathbf{S}(d)^\top$ be a left inverse
satisfying $\mathbf{S}^\dagger(d)\mathbf{S}(d) = \mathbf{I}_{K_\ast}$.
Writing $e(y)$ for the one-hot encoding of the observed joint proxy vector
$y$ and $e_a$ for the $a$ th standard basis vector in $\mathbb R^{K_\ast}$,
define
\begin{equation*}
    H_a(y,d)
    =
    e_a^\top \mathbf{S}^\dagger(d)e(y).
\end{equation*}
Then, this bridge function is unbiased for the latent variable of interest. 
\begin{align*}
    \E\!\left[
        H_a(\widetilde Y,d)
        \mid Y^\ast=b,\widetilde D=d
    \right]
    \ = \ 
    e_a^\top \mathbf{S}^\dagger(d)\mathbf{S}(d)e_b    
    \ = \ 
    e_a^\top e_b
    \ = \ 
    \bbone\{a=b\},
\end{align*}
where the first equality uses the identity $\mathbf{S}(d)e_b = \E[e(\widetilde Y) \mid Y^\ast=b,\widetilde D=d]$.
As in the binary case, these bridges need not lie in $[0,1]$ or sum to one
for each realized proxy pattern. Their role is to reproduce the latent class
indicators in conditional expectation.

To retain the same local robustness to first-stage estimation error as in the binary case, we require at least three conditionally independent measurements.
In this framework, several proxies may be combined into a composite measurement when a single proxy does not have enough categories to distinguish all $K_\ast$ latent classes.
For each proxy $j$, define $\mathbf{S}_j(d)$ analogously to $\mathbf{S}(d)$ as the conditional measurement matrix for $Y^{(j)}$, and let $\mathbf{S}_j^\dagger(d)$ denote its corresponding left inverse.
Let $M_{j,a}(d) \coloneq e_a^\top \mathbf S_j^\dagger(d) e(Y^{(j)})$.
For example, if both $Y^\ast$ and $Y^{(j)}$ are binary, then 
\begin{equation*}
    \mathbf{S}_j(d) = \begin{pmatrix}
        1 - \eta^\ast_{j,0}(d) & 1 - \eta^\ast_{j,1}(d) \\
        \eta^\ast_{j,0}(d) & \eta^\ast_{j,1}(d)
    \end{pmatrix},
\end{equation*}
and $M_{j,1}(d) =\{Y^{(j)} - \eta^\ast_{j,0}(d)\}/\{\eta^\ast_{j,1}(d)-\eta^\ast_{j,0}(d)\}$ which coincides with the single proxy bridge defined for the binary setting.
We define the robust bridge as
\begin{equation*}
    H_a^{\mathrm{R}}
    =
    3\binom{J}{2}^{-1}
    \sum_{j_1<j_2}M_{j_1,a}M_{j_2,a}
    -
    2\binom{J}{3}^{-1}
    \sum_{j_1<j_2<j_3}M_{j_1,a}M_{j_2,a}M_{j_3,a},
\end{equation*}
analogous to the binary setting.
Under conditional independence across the $J$ proxies, this bridge recovers $\bbone\{Y^\ast=a\}$ and retains the same first-order cancellation
property as the binary robust bridge. Thus, once the vector of
class-specific bridges is identified, estimation and inference proceed as in
the binary case by replacing the two-term bridge-based moment with
equation~\eqref{eq:multicategory-observed-moment}.
As discussed earlier, when $J=3$ and focusing on the latent outcome variable setting, this bridge coincides with the construction in \citet{guo2026proximal}.

The principal new requirements therefore concern identification of the measurement matrix and, consequently, the nuisance components needed to construct the bridge. 
For example, with three binary proxy labels, each proxy measurement matrix has
Kruskal rank at most two, whereas Kruskal's condition for $K_\ast$ latent classes
requires the three ranks to sum to at least $2K_\ast+2$. Multicategory latent
outcomes consequently require richer proxy supports. In addition, the anchoring rule must match all $K_\ast$ recovered
classes to their substantive meanings.

The number of response categories alone does not guarantee identification;
the corresponding class-conditional measurement matrices must also have
sufficiently independent columns. 
A simple sufficient condition is that all three matrices have full column rank $K_\ast$.
Specifically, if three conditionally
independent label groups have measurement matrices $S_j(d)$ with full column
rank $K_\ast$, then each has Kruskal rank $K_\ast$, and Kruskal's condition reduces to
$3K_\ast\geq 2K_\ast+2$, which holds for $K_\ast\geq2$. 
More generally, if two groups have full column rank $K_\ast$, the
third need only have Kruskal rank at least two. Thus, one group may have fewer
than $K_\ast$ response patterns if the other groups provide sufficient identifying
variation; for example, when $K_\ast=4$, the group support sizes $(2,4,4)$ can satisfy
the rank condition at equality when the matrices reach their maximum
Kruskal ranks. 

In the binary construction in Section~\ref{sec:latent-independent}, we showed how the rank
condition underlying classical array decomposition is related to the nonzero contrast condition required for bridge-based identification. For a binary proxy $j$, $\det\{\mathbf{S}_j(d)\} = \eta^\ast_{j,1}(d)-\eta^\ast_{j,0}(d)$.
Thus, the nonzero-contrast condition
$\eta^\ast_{j,1}(d)\neq\eta^\ast_{j,0}(d)$ is exactly the condition that
$\mathbf{S}_j(d)$ be invertible, or equivalently have Kruskal rank two.
In the multicategory setting, the corresponding requirement is that its joint measurement
matrix $\mathbf{S}_j(d)$ have full column rank. 
This full column rank requirement is sufficient to
construct a bridge; identification of the overall
latent class may hold under weaker, asymmetric Kruskal rank
conditions.
The multicategory analog that requires the binary contrast to be bounded
away from zero is a uniform lower bound on the smallest singular value, i.e., $\sigma_{\min}\{\mathbf{S}_j(d)\}$ is bounded away from zero.

\section{Designing and Assessing Conditional Independence}
\label{sec:conditional-independence-practice}

Conditional independence is a structural assumption about the residual errors
of the measurements and is not guaranteed by high predictive accuracy of each imperfect measurement.
We therefore first recommend explicitly designing multiple measurements to make the assumption
as plausible as possible (Section~\ref{subsec:more-plausible}). We then propose a diagnostic tool for the conditional independence assumption (Section~\ref{subsec:diag-without}). We finally propose a direct test of the assumption and a bias-correction estimator in a special case when a small amount of validation data is available (Section~\ref{subsec:diag-with}).

\subsection{Making the Conditional Independence Assumption More Plausible}
\label{subsec:more-plausible}

The first approach is to judiciously create a conditioning set. Unlike the traditional conditional independence assumption that only conditions on the latent variable itself (equation~\eqref{eq:classical-cond}), our method explicitly allows for conditioning on auxiliary input- or annotation-level information $D_i$ and the downstream variables $W_i$. Therefore, by incorporating rich information about each annotation task, the proposed method allows for measurement errors to be dependent through task difficulty, such as text length, complexity, language, writing style, or image quality. More generally, researchers can include any pre-specified features derived from the raw input like texts, images, and videos, such as text embeddings. One useful option is to include an estimate of task difficulty constructed from the level of agreement among additional proxies that are not used in the main analysis, e.g., quantifying the disagreement between multiple LLMs and including the disagreement score in a conditioning set.

Second, researchers can also make the conditional independence assumption more plausible by carefully constructing multiple measurements themselves. In particular, researchers should seek measurements with substantively different error mechanisms rather than selecting proxies solely by predictive accuracy. For LLM annotations, this may involve using different model families (e.g., one from the GPT family and another from an open-source model) as errors are more likely to be correlated within the same model family. In general, our method only requires $G \geq 3$ groups of proxies that are conditionally independent, but we can allow for any dependence of errors between proxies within each group. Specifically, partition the proxies into $G\geq3$ groups
$\widetilde X^{[1]},\ldots,\widetilde X^{[G]}$ and generalize
Assumption~\ref{ass:label-conditional-independence} to the following group-level conditional independence assumption.
\begin{equation*}
    \widetilde X^{[1]}
    \indep\cdots\indep
    \widetilde X^{[G]}
    \mid X^\ast,\widetilde D,
\end{equation*}
which allows for arbitrary dependence among proxies within the same group. Each group is then treated as a single multicategory measurement whose realizations are the joint label patterns. If each group has a full-rank joint
measurement matrix, one can construct a group-specific bridge and apply the robust bridge construction across groups, as in
Section~\ref{subsec:multicategory-labels}. This formulation permits, for example, annotations from GPT-4 and GPT-5 to form one group, annotations from Claude Opus and Sonnet to form another group, and those from an open-source model to form the third group, while allowing for any error dependence within each group. In practice, independently produced human annotations can also provide another useful measurement group even if they are not treated as gold-standard labels and their accuracy might be lower than LLMs. This is because human errors may differ systematically from those of LLMs. Importantly, human annotators should be blinded to the other annotations.

Finally, another strategy is to use different prompts across annotators, as a particular pattern in one prompt might create the same error across annotators. If researchers can commit to one clean and validated prompt or a codebook, that will not introduce any correlated error. But when researchers cannot pin down one particular prompt, as in many applications,  they might be able to randomly sample prompts from an admissible pool of prompts so that different annotators share fewer sources of errors. More generally, prompts may also be adapted to different types of inputs, provided that the prompt version and characteristics are recorded and included in the conditioning covariates $\widetilde D$. 

\subsection{Diagnostics Without Gold-Standard Labels}
\label{subsec:diag-without}
In most applications, researchers can also conduct a diagnostic test to evaluate the observable implications of the conditional independence assumption. 

With exactly three binary proxies and two latent classes, the unrestricted latent-class model is generically just identified at each fixed value
$\widetilde D=d$: the observed label distribution has $2^3-1=7$ degrees of
freedom, equal to the $2J+1=7$ latent-class and measurement parameters.
Consequently, the three-proxy model supplies no generic overidentifying restrictions. With $J\geq4$ proxies, however, the number of
overidentifying restrictions at each fixed $d$ is $(2^J-1)-(2J+1)=2^J-2J-2$.
These restrictions can be assessed by comparing the fitted product-mixture
distribution in equation~\eqref{eq:proxy-mixture}, with the empirical joint label distribution (i.e., $\sum_i \bbone\{\widetilde X_i = \widetilde x, \widetilde D_i = d\}/\sum_i \bbone\{\widetilde D_i = d\}$), using a
parametric bootstrap deviance or an analogous conditional moment test. 
When $\widetilde D$ is continuous or high-dimensional, cross-fitted residual
interaction moments or held-out predictive checks are more practical than
cell-by-cell tests.

A complementary, target-specific diagnostic exploits the fact that every
informative subset of at least three valid proxies identifies the same
downstream parameter. Let $\widehat\beta_S^{\mathrm{DMM}}$ denote the DMM estimate obtained
from label subset $S$. Under the identifying assumptions,
$\widehat\beta_S^{\mathrm{DMM}}$ should agree across subsets up to sampling errors. A joint
bootstrap or multiplier-bootstrap test can therefore assess heterogeneity
among the subset-specific estimates. As usual, this is simply a 
specification diagnostic rather than a direct test of conditional independence: the failure to reject the null hypothesis of no difference across the DMM estimators based on different subsets of proxies does not imply the validity of the conditional independence assumption.

\subsection{Testing and Bias-Correction with Small Gold-Standard Data}
\label{subsec:diag-with}
When the conditional independence assumption is in serious doubt, it is generally recommended to collect the gold-standard labels, even for a small amount of data, for two reasons: (a) direct testing of the conditional independence assumption and (b) bias correction without assuming the conditional independence assumption. 

First, a gold-standard sample permits a direct assessment of the
conditional independence assumption. Let $R$ indicate that $X^\ast$ is observed and let
$\rho$ be its known sampling probability, as in Section~\ref{sec:problem-setting}.
For a subset of labels $S$ with $|S|\geq2$, conditional independence implies
the residual moment restriction
\begin{equation*}
    \E\left[
        \frac{R}{\rho}
        \bbone\{X^\ast=a\}
        h(\widetilde D)
        \prod_{j\in S}
        \left\{
            X^{(j)}-\eta^\ast_{j,a}(\widetilde D)
        \right\}
    \right]
    =0,
    \qquad a\in\{0,1\},
\end{equation*}
for suitable instrument functions $h$. Pair and higher-order residual products
can be tested jointly using cross-fitted estimates of the class-specific label
means. With low-dimensional discrete covariates, the same idea can be
implemented using stratified log-linear or permutation tests. Again, because a small validation sample may have limited power, failure to reject these tests should not be interpreted as establishing the conditional independence assumption.

Second, a small validation sample can also be used to combine DMM and methods that use validation data, such as DSL and PPI, to directly bias-correct downstream inference. The DSL correction can be
applied to any cross-fitted imputation of the full-data moment function,
not only one obtained from a supervised prediction of the latent label. We
therefore use the DMM moment function as an imputation in the DSL framework.
Let $\widehat\psi_i^{\mathrm{DMM}}(\beta)$ 
denote the cross-fitted DMM moment function, which is available for every
unit, and recall that $\psi^F(Y_i,X_i^\ast,W_i;\beta)$
denotes the full-data moment function, which is observed only when $R_i=1$.
The resulting DMM-assisted DSL moment function is
\begin{equation}
    \widehat\psi_i^{\mathrm{DMM\text{-}DSL}}(\beta)
    =
    \widehat\psi_i^{\mathrm{DMM}}(\beta)
    +
    \frac{R_i}{\rho_i}
    \left\{
        \psi^F(Y_i,X_i^\ast,W_i;\beta)
        -
        \widehat\psi_i^{\mathrm{DMM}}(\beta)
    \right\}.
    \label{eq:validation-augmented-dmm}
\end{equation}
Thus, DMM imputes the full-data moment function for all observations, and
the validation observations estimate and correct any remaining imputation
error.

Under the known sampling design for obtaining gold-standard labels without assuming conditional independence of multiple imperfect measurements,
\begin{equation*}
    \E\left[
        \widehat\psi_i^{\mathrm{DMM\text{-}DSL}}(\beta)
        \mid
        X_i^\ast,\widetilde{D}_i,\widetilde X_i,
        \mathcal I_{-k(i)}
    \right]
    =
    \psi^F(Y_i,X_i^\ast,W_i;\beta),
\end{equation*}
where $\mathcal I_{-k(i)}$ denotes the sample used to construct the
cross-fitted DMM nuisance estimates. This identity does not require the DMM
conditional-independence model to be correct. The validation design therefore
provides validity, while DMM serves as an imputation or control variate that
can improve precision when its bridge-based moment function is informative.

More generally, the DMM-assisted DSL moment can be written as the
following control-variate adjustment (or power-tuning version in \citealp{angelopoulos2023ppi++}):
\begin{equation}
    \widehat\psi_i^{\mathrm{DMM\text{-}DSL}}(\beta;C)
    =
    \underbrace{
        \frac{R_i}{\rho_i}
        \psi^F(Y_i,X_i^\ast,W_i;\beta)
    }_{\text{gold-standard-only moment}}
    -
    C
    \underbrace{
        \left(\frac{R_i}{\rho_i}-1\right)
        \widehat\psi_i^{\mathrm{DMM}}(\beta)
    }_{\text{mean-zero control variate}},
    \label{eq:dmm-dsl-power-family}
\end{equation}
where $C$ is a scalar for a scalar moment or a conformable matrix for a
vector-valued moment.
The coefficient $C$ may be estimated on separate folds to reduce the
asymptotic variance. Setting $C=0$ yields gold-standard-only estimation,
whereas setting $C=\mathbf I$ yields
equation~\eqref{eq:validation-augmented-dmm}.

\section{Simulation and Empirical Validation}
\label{sec:applications}

We use simulation studies (Section~\ref{sec:fowler-application}) and real-world empirical validation (Section~\ref{subsec:pan-chen}) to evaluate DMM. First, a simulation study calibrated based on an empirical study of political advertisements 
\citep{fowler2021political} examines a latent dependent variable under a
synthetic data-generating process that satisfies the conditional independence assumption. This simulation study evaluates the finite-sample performance of the DMM estimator under the required assumption. 
Second, an empirical
validation based on \citet{pan2018concealing} examines a latent independent
variable using LLM annotations. This is an empirical validation study in that we do not impose the conditional independence assumption, and we empirically evaluate whether DMM without access to the gold-standard labels can recover the oracle benchmark estimate that relies on gold-standard labels. More details about each application are provided in Appendix~\ref{app:application-details}.

\subsection{Simulation Study: Platform Differences in Political Advertising Tone}
\label{sec:fowler-application}

\citet{fowler2021political} compare the tone of political advertisements on
Facebook and television. We use the expert-coded ``Promote'' indicator as the
latent dependent variable and target the coefficient on the Facebook indicator
in an ad-level logistic regression adjusting for party, incumbency status, and
office type. We hold the observed covariates for $12{,}973$ advertisements
fixed and simulate the latent outcome and proxy label models so that the resulting distributions are similar to the distribution of the original expert labels and LLM annotations. Please see Section~5 of
\citet{egami2024using} for details on these data. In each simulation run, we generate imperfect measurements independently conditional on the latent outcome and covariates so that the conditional independence assumption holds by construction and the conditional classification rate model used by DMM is correctly specified. The target Facebook
coefficient is $1.314$, and we conduct $500$ Monte Carlo replications.

To evaluate the impact of increasing the number of proxies, we rank the proxy labels by their F1 scores against the original expert labels and construct nested proxy sets by adding labels one at a time in decreasing order of F1 score. For each number of labels $J$, the naive estimator replaces
the latent outcome with the majority vote among the $J$ proxy labels. When
$J$ is even, ties are resolved by random draw. DMM instead combines the multiple imperfect labels
using the robust bridge function for downstream inference. We compare both estimators with the infeasible oracle estimator that assumes access to gold-standard labels. We report bias, root mean squared error (RMSE), and empirical coverage of nominal $95\%$ confidence intervals.

Figure~\ref{fig:fowler-promote-proxy-count} shows the results. As repeatedly found in the literature, ignoring measurement errors leads to substantial bias even when the individual proxy labels have high
F1 scores. With the top three proxies, the naive majority-vote estimator has
bias $0.186$, coverage $0.018$, and RMSE $0.191$. However, using the same three imperfect labels, DMM can explicitly provide debiased inference: 
DMM has bias $0.002$, coverage $0.948$, and RMSE $0.063$, compared with an
RMSE of $0.045$ for the oracle benchmark estimate. Thus, DMM reduces RMSE by approximately two
thirds relative to majority voting and restores the valid confidence interval (i.e., the empirical coverage is nearly the nominal level).

Figure~\ref{fig:fowler-promote-proxy-count} also shows how DMM behaves as additional proxies are introduced. Across the evaluated proxy
sets, DMM's bias remains between $0.002$ and $0.003$, and its coverage remains
between $0.946$ and $0.962$. By contrast, the majority-vote estimator has
bias between $0.085$ and $0.226$, with coverage ranging from $0$ to $0.538$.
Even with $15$ proxy labels, majority voting has bias $0.085$, coverage
$0.538$, and RMSE $0.096$, while DMM has bias $0.002$, coverage $0.946$,
and RMSE $0.047$.

\begin{figure}[t]
    \centering
    \includegraphics[width=\linewidth]
    {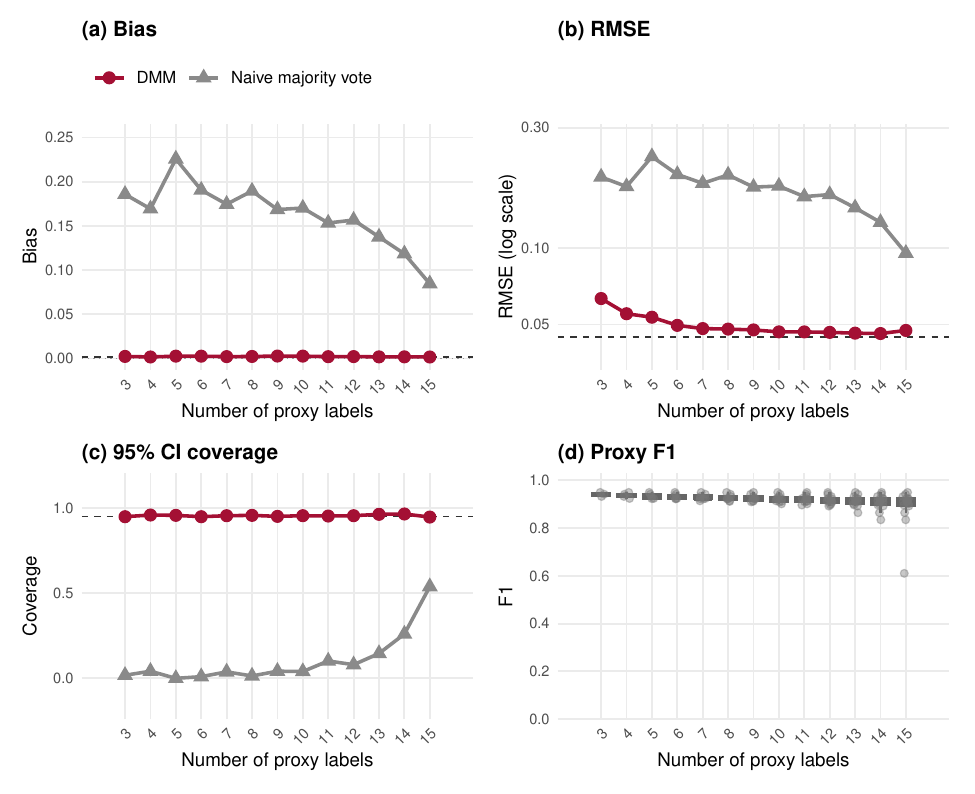}
    \caption{
    Simulation comparing DMM with naive majority voting as the
    number of proxy labels increases. Panels~(a)--(c) report bias, RMSE, and
    $95\%$ confidence-interval coverage over $500$ replications; Panel~(d)
    reports proxy F1 scores. Dashed lines mark the oracle benchmarks.
    }
    \label{fig:fowler-promote-proxy-count}
\end{figure}

Additional informative labels also improve the precision of DMM. Its RMSE
falls from $0.063$ with three labels to $0.053$, $0.048$, and $0.047$ with
the top 5, 7, and 10 labels, respectively, and reaches $0.046$ with 13
labels, close to the oracle RMSE of $0.045$. These gains eventually flatten
as progressively weaker labels are added. Thus, additional sufficiently
informative labels can narrow the precision gap between DMM and the oracle infeasible benchmark estimator that observes the latent variable for all units, but not every additional label is beneficial.

\subsection{Empirical Validation: Accusations of Wrongdoing in China}
\label{subsec:pan-chen}
\citet{pan2018concealing} study whether Chinese officials conceal online
complaints that accuse local officials of wrongdoing. We treat the
expert-coded indicator of prefecture-level wrongdoing as the latent independent
variable and upward reporting as the downstream dependent variable. 
The main estimand is the coefficient on prefecture-level wrongdoing in a logistic regression of upward reporting on the latent indicator and a subset of the covariates ($7$ binary and $1$ continuous) used in the original application.
The analysis contains $1{,}412$ complaints and uses the LLM annotations as proxy labels. We use three LLM labels in this study: GPT-4.1 5-shot, GPT-4 5-shot, and Llama-4 0-shot. 

We compare (a) the DMM estimator using the robust bridge without access to the gold-standard labels, (b) three naive
regressions that replace the latent variable with one proxy at a time, and (c) DSL using a simple random sample of $500$ gold-standard labels ($35\%$ of the sample). We evaluate them against the oracle logistic regression that uses the gold-standard labels for all units. Importantly, in this empirical validation, we use the real-world data and do not simulate any data. Therefore, this is a realistic evaluation of the DMM estimator because we do not know whether the conditional independence assumption required in the DMM estimator holds, as in the real-world empirical application.

Panel~(a) of Figure~\ref{fig:pan-chen-empirical-results} reports the point estimates and $95\%$ confidence intervals. The gold-standard benchmark for the
prefecture-wrongdoing coefficient is $-1.039$. Despite their similarly high
reported class-weighted F1 scores, the three naive estimates range from
$-1.869$ to $-0.388$ and all of them are biased. DSL and DMM estimate the coefficient to be $-1.279$ and $-0.844$, respectively, and both confidence intervals contain the benchmark estimate. Importantly, unlike DSL, DMM obtains its estimate without using any gold-standard labels.

\begin{figure}[t]
    \centering
    \includegraphics[width=\linewidth]
    {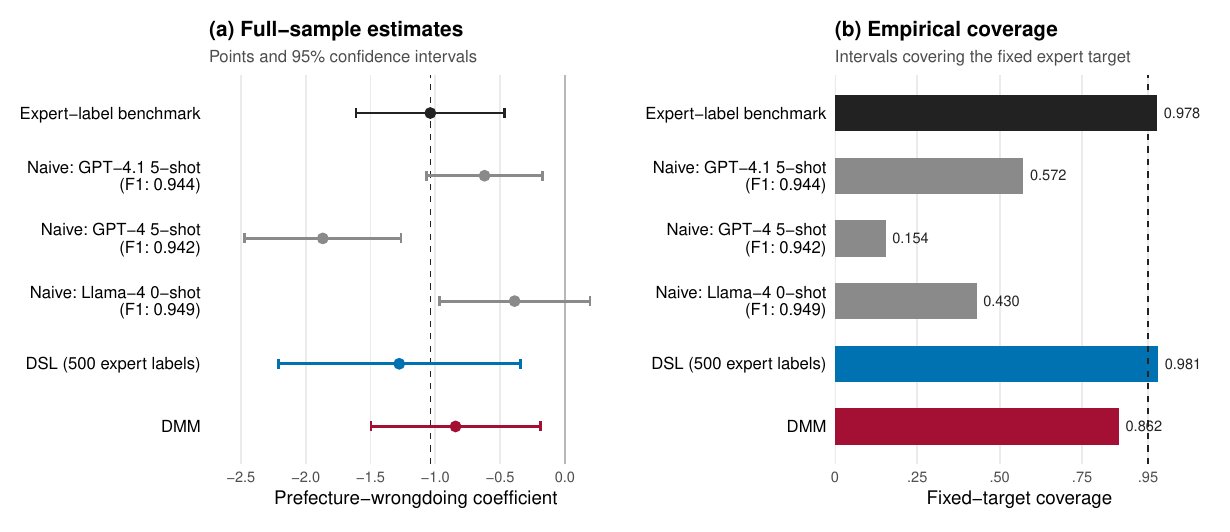}
    \caption{
    \cite{pan2018concealing} empirical estimates and coverage.
Panel~(a) reports coefficient estimates and 95\% confidence intervals; the
dashed line marks the full-sample expert-label benchmark. DMM uses no gold-standard labels, whereas DSL uses $500$. Panel~(b) reports the proportion of intervals
covering the fixed expert benchmark over $500$ bootstrap resamples.
    }
    \label{fig:pan-chen-empirical-results}
\end{figure}

Panel~(b) reports empirical coverage based on
$500$ resamples. Specifically, it shows the proportion of $95\%$
confidence intervals that contain the benchmark estimate.
The three naive estimators have coverage rates ranging from $0.154$ to $0.572$, whereas DSL and DMM have coverage rates of $0.981$ and $0.862$, respectively.

\section{Discussion}
\label{sec:discussion}

In this paper, we proposed debiased inference with multiple imperfect measurements (DMM), a framework to conduct valid downstream inference without gold-standard labels. Under the conditional independence assumption given the latent variable and any observed features of annotation tasks, such as annotation difficulty, captured by text embeddings, DMM identifies the conditional classification rate model required for valid downstream inference. 
Importantly, individual measurements may remain systematically biased, differ substantially in accuracy, and exhibit error rates that vary across units. By allowing the (mis)classification rate model to depend on input- and annotation-level information, our framework makes the conditional independence assumption more plausible and its substantive implications more transparent. We also provide practical guidance for designing measurement procedures that support this assumption and for assessing its plausibility in applied settings.

Our main methodological contribution is to connect the multiple measurement model to general downstream analyses covering both settings in which the latent variable is either an independent or a dependent variable. 
We construct a robust bridge whose conditional expectation recovers the latent indicator and use it to reproduce the oracle full-data moment function. By carefully designing how the multiple measurements are combined within the bridge, we show that the resulting construction is robust to misspecification of any one label-specific nuisance model.
Combined with cross-fitting, this yields a consistent and asymptotically normal estimator under standard product-rate conditions, allowing the nuisance estimators to converge at slow nonparametric rates $n^{-1/4}$ while preserving root-$n$ inference for the downstream parameter.
We also discuss extensions to multicategory variables under the corresponding rank conditions.

Our method accommodates any number of imperfect measurements, from as few as three to much larger collections, allowing researchers to adapt the measurement design to different annotation budgets and settings. 
The framework also clarifies how the number and quality of the available measurements affect statistical precision. 
Adding sufficiently informative measurements can reduce the measurement-induced component of the asymptotic variance, bringing DMM closer in efficiency to the oracle estimator. However, precision does not improve mechanically with the number of measurements: low-quality measurements can yield unstable bridges and offset the gains from averaging over a larger collection.

We present simulation and empirical validation studies to demonstrate the performance of DMM. Using Monte Carlo simulations calibrated to data from \citet{fowler2021political}, we show that DMM exhibits small bias and achieves nominal coverage across different numbers of labels, whereas naive estimators directly using imperfect labels remain biased and exhibit undercoverage even when individual proxies have high F1 scores. The empirical validation study based on \citet{pan2018concealing} provides a complementary illustration using real LLM annotations. DMM produces an estimate and confidence interval close to the oracle benchmark without using expert labels in settings where the conditional independence assumption is not guaranteed to hold, as in many applications. 

The main limitation of DMM is that it requires the conditional independence assumption about measurement errors, even though we made it much more plausible than the classical counterpart by allowing for a rich conditioning set. Conditional independence is not guaranteed by high classification accuracy and may be violated when annotations interact with omitted features of the input. We therefore devote substantial attention to how these concerns can be mitigated and assessed in practice. 
In Section~\ref{sec:conditional-independence-practice}, we discuss designing measurement systems with diverse annotators (either human or machine), conditioning on input- and annotation-level information that may explain common errors, and examining subset stability and other observable implications of the fitted measurement model. 
When a limited validation sample is available, it can additionally be used to assess the conditional independence assumption and to combine DMM with the design-based bias correction of DSL. 

More broadly, DMM and validation-based approaches should be viewed as complementary strategies. DMM is particularly useful when multiple measurements are expected to have reasonably high accuracy while they still have non-random, nonclassical measurement errors. Even when measurement errors are small, unless corrected carefully with DMM, they can induce substantial bias and invalidate downstream inference. Bias-correction methods using validation data, such as DSL and PPI, are preferable when gold-standard labels can be collected, and the conditional independence assumption is less credible even after conditioning on rich annotation-level variables.

\clearpage
\bibliography{bibliography}

@article{yang2026benchmark,
  title={{Benchmark Illusion: Disagreement among LLMs and Its Scientific Consequences}},
  author={Yang, Eddie and Wang, Dashun},
  journal={arXiv preprint arXiv:2602.11898},
  year={2026}
}

@article{spirling2023open,
  title={{Why Open-Source Generative AI Models Are An Ethical Way Forward For Science}},
  author={Spirling, Arthur},
  journal={Nature},
  volume={616},
  number={7957},
  pages={413--413},
  year={2023},
  publisher={Nature}
}

@article{baumann2025large,
  title={{Large Language Model Hacking: Quantifying the Hidden Risks of Using LLMs for Text Annotation}},
  author={Baumann, Joachim and R{\"o}ttger, Paul and Urman, Aleksandra and Wendsj{\"o}, Albert and Plaza-del-Arco, Flor Miriam and Gruber, Johannes B and Hovy, Dirk},
  journal={arXiv preprint arXiv:2509.08825},
  year={2025}
}

@article{wang2020methods,
  title={{Methods for Correcting Inference Based on Outcomes Predicted by Machine Learning}},
  author={Wang, Siruo and McCormick, Tyler H and Leek, Jeffrey T},
  journal={Proceedings of the National Academy of Sciences},
  volume={117},
  number={48},
  pages={30266--30275},
  year={2020},
  publisher={National Academy of Sciences}
}

@article{egami2024identification,
  title={{Identification and Estimation of Causal Peer Effects Using Double Negative Controls for Unmeasured Network Confounding}},
  author={Egami, Naoki and Tchetgen Tchetgen, Eric J},
  journal={Journal of the Royal Statistical Society Series B: Statistical Methodology},
  volume={86},
  number={2},
  pages={487--511},
  year={2024},
  publisher={Oxford University Press US}
}

@techreport{hansen2023remote,
  title={{Remote Work across Jobs, Companies, and Space}},
  author={Hansen, Stephen and Lambert, Peter John and Bloom, Nicholas and Davis, Steven J and Sadun, Raffaella and Taska, Bledi},
  year={2023},
  institution={National Bureau of Economic Research}
}

@article{halterman2026codebook,
  title={{Codebook LLMs: Evaluating LLMs as Measurement Tools for Political Science Concepts}},
  author={Halterman, Andrew and Keith, Katherine A},
  journal={Political Analysis},
  volume={34},
  number={2},
  pages={188--204},
  year={2026},
  publisher={Cambridge University Press}
}

@article{ziems2024can,
  title={{Can Large Language Models Transform Computational Social Science?}},
  author={Ziems, Caleb and Held, William and Shaikh, Omar and Chen, Jiaao and Zhang, Zhehao and Yang, Diyi},
  journal={Computational Linguistics},
  volume={50},
  number={1},
  pages={237--291},
  year={2024},
  publisher={MIT Press One Broadway, 12th Floor, Cambridge, Massachusetts 02142, USA~…}
}

@article{gilardi2023chatgpt,
  title={{ChatGPT Outperforms Crowd Workers for Text-Annotation Tasks}},
  author={Gilardi, Fabrizio and Alizadeh, Meysam and Kubli, Ma{\"e}l},
  journal={Proceedings of the National Academy of Sciences},
  volume={120},
  number={30},
  pages={e2305016120},
  year={2023},
  publisher={National Academy of Sciences}
}

@article{chen2000unified,
  title={{A Unified Approach to Regression Analysis under Double-Sampling Designs}},
  author={Chen, Yi-Hau and Chen, Hung},
  journal={Journal of the Royal Statistical Society: Series B (Statistical Methodology)},
  volume={62},
  number={3},
  pages={449--460},
  year={2000},
  publisher={Oxford University Press}
}

@book{van2000asymptotic,
  title={Asymptotic Statistics},
  author={Van der Vaart, Aad W},
  volume={3},
  year={2000},
  publisher={Cambridge University Press}
}

@article{vansteelandt2022assumption,
  title={{Assumption-lean Inference for Generalised Linear Model Parameters}},
  author={Vansteelandt, Stijn and Dukes, Oliver},
  journal={Journal of the Royal Statistical Society: Series B (Statistical Methodology)},
  volume={84},
  number={3},
  pages={657--685},
  year={2022},
  publisher={Wiley Online Library}
}

@article{buja2019models,
  title={{Models as Approximations I}},
  author={Buja, Andreas and Brown, Lawrence and Berk, Richard and George, Edward and Pitkin, Emil and Traskin, Mikhail and Zhang, Kai and Zhao, Linda},
  journal={Statistical Science},
  volume={34},
  number={4},
  pages={523--544},
  year={2019},
  publisher={JSTOR}
}

@article{chen2026partial,
  title={{Partial Identification from LLM Prompts}},
  author={Chen, Xiaohong and Rambachan, Ashesh and Tamer, Elie},
  journal={arXiv preprint arXiv:2606.15031},
  year={2026}
}

@article{kim2025correlated,
  title={{Correlated Errors in Large Language Models}},
  author={Kim, Elliot and Garg, Avi and Peng, Kenny and Garg, Nikhil},
  journal={arXiv preprint arXiv:2506.07962},
  year={2025}
}

@article{cher2018double,
	Author = {Chernozhukov, Victor and Chetverikov, Denis and Demirer, Mert and Duflo, Esther and Hansen, Christian and Newey, Whitney and Robins, James},
	Journal = {Econometrics Journal},
	Pages = {C1 -- C68},
	Title = {{Double/Debiased Machine Learning for Treatment and Structural Parameters}},
	Volume = {21},
	Year = {2018}}

@article{robins1994,
	Author = {Robins, James M and Rotnitzky, Andrea and Zhao, Lue Ping},
	Journal = {Journal of the American Statistical Association},
	Number = {427},
	Pages = {846--866},
	Publisher = {Taylor \& Francis},
	Title = {{Estimation of Regression Coefficients When Some Regressors Are Not Always Observed}},
	Volume = {89},
	Year = {1994}}

@article{chen2008semiparametric,
  title={{Semiparametric Efficiency in GMM Models with Auxiliary Data}},
  author={Chen, Xiaohong and Hong, Han and Tarozzi, Alessandro},
  year={2008},
  journal={Annals of Statistics},
}

@article{ludwig2024large,
  title={{Large Language Models: An Applied Econometric Framework}},
  author={Ludwig, Jens and Mullainathan, Sendhil and Rambachan, Ashesh},
  journal={Annual Review of Economics},
  volume={18},
  year={2026},
  publisher={Annual Reviews}
}

@article{carlson2025unifying,
  title={{A Unifying Framework for Robust and Efficient Inference with Unstructured Data}},
  author={Carlson, Jacob and Dell, Melissa},
  journal={arXiv preprint arXiv:2505.00282},
  year={2025}
}

@article{hu2008identification,
  title={Identification and estimation of nonlinear models with misclassification error using instrumental variables: A general solution},
  author={Hu, Yingyao},
  journal={Journal of Econometrics},
  volume={144},
  number={1},
  pages={27--61},
  year={2008},
  publisher={Elsevier}
}

@unpublished{bouyamourn2026interpretable,
  author = {Bouyamourn, Adam and Spirling, Arthur},
  title  = {Interpretable Aggregation of Correlated {LLM} Annotations},
  year   = {2026},
  month  = jul,
  note   = {Manuscript presented at the 2026 Annual Meeting of the Society for Political Methodology}
}

@article{katsumata2023statistical,
  author      = {Katsumata, Hiroto and Yamauchi, Soichiro},
  title       = {Statistical Analysis with Machine Learning Predicted Variables},
  journal     = {Working Paper},
  year        = {2023}
}

@inproceedings{mozer2023decreasing,
  title={Decreasing the human coding burden in randomized trials with text-based outcomes via model-assisted impact analysis},
  author={Mozer, Reagan and Miratrix, Luke},
  booktitle={2023 IMS International Conference on Statistics and Data Science (ICSDS)},
  pages={520},
  year={2023}
}

@article{angelopoulos2023prediction,
  title={Prediction-powered inference},
  author={Angelopoulos, Anastasios N and Bates, Stephen and Fannjiang, Clara and Jordan, Michael I. and Zrnic, Tijana},
  journal={Science},
  volume={382},
  number={6671},
  pages={669--674},
  year={2023},
  publisher={American Association for the Advancement of Science}
}

@article{egami2024using,
  title={Using large language model annotations for the social sciences: A general framework of using predicted variables in downstream analyses},
  author={Egami, Naoki and Hinck, Musashi and Stewart, Brandon M and Wei, Hanying},
  journal={American Journal of Political Science},
  year={2026}
}

@article{egami2023using,
  title={Using imperfect surrogates for downstream inference: Design-based supervised learning for social science applications of large language models},
  author={Egami, Naoki and Hinck, Musashi and Stewart, Brandon and Wei, Hanying},
  journal={Advances in Neural Information Processing Systems},
  volume={36},
  pages={68589--68601},
  year={2023}
}

@article{nakamura2025surrogate,
  title={Surrogate Representation Inference for Text and Image Annotations},
  author={Nakamura, Kentaro},
  journal={arXiv preprint arXiv:2509.12416},
  year={2025}
}

@article{kruskal1977three,
  title={Three-way arrays: rank and uniqueness of trilinear decompositions, with application to arithmetic complexity and statistics},
  author={Kruskal, Joseph B},
  journal={Linear Algebra and its Applications},
  volume={18},
  number={2},
  pages={95--138},
  year={1977},
  publisher={Elsevier}
}

@article{schennach2016recent,
  title={Recent advances in the measurement error literature},
  author={Schennach, Susanne M},
  journal={Annual Review of Economics},
  volume={8},
  number={1},
  pages={341--377},
  year={2016},
  publisher={Annual Reviews}
}

@article{schennach2022measurement,
  title={Measurement systems},
  author={Schennach, Susanne},
  journal={Journal of Economic Literature},
  volume={60},
  number={4},
  pages={1223--1263},
  year={2022},
  publisher={American Economic Association 2014 Broadway, Suite 305, Nashville, TN 37203-2425}
}

@article{hu2008instrumental,
  title={Instrumental variable treatment of nonclassical measurement error models},
  author={Hu, Yingyao and Schennach, Susanne M},
  journal={Econometrica},
  volume={76},
  number={1},
  pages={195--216},
  year={2008},
  publisher={Wiley Online Library}
}

@article{kuroki2014measurement,
  title={Measurement bias and effect restoration in causal inference},
  author={Kuroki, Manabu and Pearl, Judea},
  journal={Biometrika},
  pages={423--437},
  year={2014},
  publisher={JSTOR}
}

@article{zhou2024causal,
  title={Causal inference for a hidden treatment},
  author={Zhou, Ying and Tchetgen Tchetgen, Eric},
  journal={arXiv preprint arXiv:2405.09080},
  year={2024}
}

@article{miao2018identifying,
  title={Identifying causal effects with proxy variables of an unmeasured confounder},
  author={Miao, Wang and Geng, Zhi and Tchetgen Tchetgen, Eric J},
  journal={Biometrika},
  volume={105},
  number={4},
  pages={987--993},
  year={2018},
  publisher={Oxford University Press}
}

@article{allman2009identifiability,
  title   = {Identifiability of Parameters in Latent Structure Models with Many Observed Variables},
  author  = {Allman, Elizabeth S. and Matias, Catherine and Rhodes, John A.},
  journal = {The Annals of Statistics},
  year    = {2009},
  volume  = {37},
  number  = {6A},
  pages   = {3099--3132},
  doi     = {10.1214/09-AOS689}
}

@article{chernozhukov2022locally,
  title={Locally robust semiparametric estimation},
  author={Chernozhukov, Victor and Escanciano, Juan Carlos and Ichimura, Hidehiko and Newey, Whitney K and Robins, James M},
  journal={Econometrica},
  volume={90},
  number={4},
  pages={1501--1535},
  year={2022},
  publisher={Wiley Online Library}
}

@article{guo2026proximal,
  title={Proximal Causal Inference for Hidden Outcomes},
  author={Guo, Helen and Ghassami, AmirEmad and Shpitser, Ilya and Ogburn, Elizabeth L},
  journal={arXiv preprint arXiv:2605.09849},
  year={2026}
}

@article{fowler2021political,
  title={Political Advertising Online and Offline},
  author={Fowler, Erika Franklin and Franz, Michael M. and Martin, Gregory J. and Peskowitz, Zachary and Ridout, Travis N.},
  journal={American Political Science Review},
  volume={115},
  number={1},
  pages={130--149},
  year={2021}
}

@article{pan2018concealing,
  title={Concealing Corruption: How Chinese Officials Distort Upward Reporting of Online Grievances},
  author={Pan, Jennifer and Chen, Kaiping},
  journal={American Political Science Review},
  volume={112},
  number={3},
  pages={602--620},
  year={2018}
}

@article{battaglia2024inference,
  author  = {Battaglia, Laura and Christensen, Timothy and Hansen, Stephen and Sacher, Szymon},
  title   = {Inference for Regression with Variables Generated by {AI} or Machine Learning},
  journal = {arXiv preprint arXiv:2402.15585},
  year    = {2024},
  doi     = {10.48550/arXiv.2402.15585}
}

@article{dawid1979maximum,
  title={Maximum likelihood estimation of observer error-rates using the EM algorithm},
  author={Dawid, Alexander Philip and Skene, Allan M},
  journal={Journal of the Royal Statistical Society: Series C (Applied Statistics)},
  volume={28},
  number={1},
  pages={20--28},
  year={1979},
  publisher={Wiley Online Library}
}

@article{chen2011nonlinear,
  title={Nonlinear models of measurement errors},
  author={Chen, Xiaohong and Hong, Han and Nekipelov, Denis},
  journal={Journal of Economic Literature},
  volume={49},
  number={4},
  pages={901--937},
  year={2011},
  publisher={American Economic Association}
}

@article{goodman1974exploratory,
  title={Exploratory latent structure analysis using both identifiable and unidentifiable models},
  author={Goodman, Leo A},
  journal={Biometrika},
  volume={61},
  number={2},
  pages={215--231},
  year={1974},
  publisher={Oxford University Press}
}

@article{hall2003nonparametric,
  title={Nonparametric estimation of component distributions in a multivariate mixture},
  author={Hall, Peter and Zhou, Xiao-Hua},
  journal={The Annals of Statistics},
  volume={31},
  number={1},
  pages={201--224},
  year={2003},
  publisher={Institute of Mathematical Statistics}
}

@article{bonhomme2016estimating,
  title={Estimating multivariate latent-structure models},
  author={Bonhomme, St{\'e}phane and Jochmans, Koen and Robin, Jean-Marc},
  journal={The Annals of Statistics},
  volume={44},
  number={2},
  pages={540--563},
  year={2016}
}

@article{raykar2010learning,
  title={Learning from crowds.},
  author={Raykar, Vikas C and Yu, Shipeng and Zhao, Linda H and Valadez, Gerardo Hermosillo and Florin, Charles and Bogoni, Luca and Moy, Linda},
  journal={Journal of Machine Learning Research},
  volume={11},
  number={43},
  pages={1297--1322},
  year={2010}
}

@article{zhang2016spectral,
  title={Spectral methods meet EM: A provably optimal algorithm for crowdsourcing},
  author={Zhang, Yuchen and Chen, Xi and Zhou, Dengyong and Jordan, Michael I.},
  journal={Journal of Machine Learning Research},
  volume={17},
  number={102},
  pages={1--44},
  year={2016}
}

@article{anandkumar2014tensor,
  title={Tensor decompositions for learning latent variable models},
  author={Anandkumar, Animashree and Ge, Rong and Hsu, Daniel and Kakade, Sham M and Telgarsky, Matus},
  journal={Journal of Machine Learning Research},
  volume={15},
  number={1},
  pages={2773--2832},
  year={2014},
  publisher={JMLR. org}
}

@inproceedings{chaudhuri2009multi,
  title={Multi-view clustering via canonical correlation analysis},
  author={Chaudhuri, Kamalika and Kakade, Sham M. and Livescu, Karen and Sridharan, Karthik},
  booktitle={Proceedings of the 26th Annual International Conference on Machine Learning},
  pages={129--136},
  year={2009}
}

@article{stegeman2007kruskal,
  title={On Kruskal’s uniqueness condition for the Candecomp/Parafac decomposition},
  author={Stegeman, Alwin and Sidiropoulos, Nicholas D},
  journal={Linear Algebra and its Applications},
  volume={420},
  number={2-3},
  pages={540--552},
  year={2007},
  publisher={Elsevier}
}

@article{shen1997methods,
  title={On methods of sieves and penalization},
  author={Shen, Xiaotong},
  journal={The Annals of Statistics},
  pages={2555--2591},
  year={1997},
  publisher={JSTOR}
}

@article{chen2012estimation,
  title={Estimation of nonparametric conditional moment models with possibly nonsmooth generalized residuals},
  author={Chen, Xiaohong and Pouzo, Demian},
  journal={Econometrica},
  volume={80},
  number={1},
  pages={277--321},
  year={2012},
  publisher={Wiley Online Library}
}

@article{dempster1977maximum,
  title={Maximum likelihood from incomplete data via the EM algorithm},
  author={Dempster, Arthur P and Laird, Nan M and Rubin, Donald B},
  journal={Journal of the Royal Statistical Society: Series B (Statistical Methodology)},
  volume={39},
  number={1},
  pages={1--38},
  year={1977},
  publisher={Wiley Online Library}
}

@article{carroll1970analysis,
  title={Analysis of individual differences in multidimensional scaling via an N-way generalization of “Eckart-Young” decomposition},
  author={Carroll, J Douglas and Chang, Jih-Jie},
  journal={Psychometrika},
  volume={35},
  number={3},
  pages={283--319},
  year={1970},
  publisher={Cambridge University Press}
}

@article{harshman1970foundations,
  title={Foundations of the PARAFAC procedure: Models and conditions for an ``explanatory'' multi-modal factor analysis},
  author={Harshman, Richard A},
  journal={UCLA working papers in phonetics},
  volume={16},
  number={1},
  pages={84},
  year={1970},
  publisher={Los Angeles, CA}
}

@article{angelopoulos2023ppi++,
  title={PPI++: Efficient prediction-powered inference},
  author={Angelopoulos, Anastasios N and Duchi, John C and Zrnic, Tijana},
  journal={arXiv preprint arXiv:2311.01453},
  year={2023}
}

\clearpage
\appendix

\numberwithin{equation}{section}
\counterwithin{figure}{section}
\counterwithin{table}{section}

\section{Proofs}
\label{app:proof}

Throughout this section, partition $\{1,\ldots,n\}$ into $K$ folds
$\mathcal I_1,\ldots,\mathcal I_K$. For each fold $k$, let
$\mathcal I_{-k}=\{1,\ldots,n\}\setminus\mathcal I_k$ denote the corresponding
training indices and let $n_k=|\mathcal I_k|$, where
$O_i=(Y_i,\widetilde X_i,\widetilde D_i)$. Nuisance estimates indexed by
$(-k)$ are fitted using observations indexed by $\mathcal I_{-k}$. We take the
number of folds and the fold proportions to be fixed. For any measurable function $f$, write
\begin{equation*}
    Pf=\E\{f(O)\},
    \qquad
    \mathbb P_{n,k}f
    =
    \frac{1}{n_k}\sum_{i\in\mathcal I_k}f(O_i),
    \qquad
    \mathbb P_nf
    =
    \frac{1}{n}\sum_{i=1}^n f(O_i).
\end{equation*}
For a generic nuisance collection $\eta \coloneq \{\eta_{j,0}(\cdot), \eta_{j,1}(\cdot)\}_{j=1}^{J}$, let
\begin{equation*}
    \psi^{\mathrm{DMM}}(O;\beta,\eta)
    \ = \ 
    \{1-H^{\mathrm{R}}(\widetilde{X}, \widetilde{D};\eta)\}\psi_0(Y,W;\beta)
    \ + \ 
    H^{\mathrm{R}}(\widetilde{X}, \widetilde{D};\eta)\psi_1(Y,W;\beta)
\end{equation*}
where $H^{\mathrm{R}}(\widetilde{X}, \widetilde D; \eta^\ast) = 3\overline H_2(\widetilde{X}, \widetilde D; \eta^\ast)-2\overline H_3(\widetilde{X}, \widetilde D; \eta^\ast)$ and $\psi_a(Y,W;\beta) \coloneq \psi^F(Y,X^\ast = a, W;\beta).$ The fold-specific fitted objects are denoted by $\widehat M_{j,-k}=M_j(\widetilde{D};\widehat\eta_{-k})$ and
$\widehat H_{-k}^{\mathrm{R}}=H^{\mathrm{R}}(\widetilde{X}, \widetilde{D};\widehat\eta_{-k})$,
whereas $M_{j,0}=M_j(\widetilde{D};\eta^\ast)$ and
$H_0^{\mathrm{R}}=H^{\mathrm{R}}(\widetilde{X}, \widetilde{D};\eta^\ast)$ denote their population
counterparts.

\subsection{Proof of Proposition~\ref{prop:robust-symmetric-bridge}: Bounds for the Robust Bridge}
\label{app:proof-rs-bridge}

For $a\in\{0,1\}$, define the fold-specific single-bridge error
\begin{equation*}
    e_{j,a}^{(-k)}(\widetilde D)
    =
    \E\!\left[
        \widehat M_{j,-k}
        \mid X^\ast=a,\widetilde D,\mathcal I_{-k}
    \right]
    -a,
\end{equation*}
and the corresponding robust bridge error
\begin{equation*}
    B_{a,-k}^{\mathrm{R}}(\widetilde D)
    =
    \E\!\left[
        \widehat H_{-k}^{\mathrm{R}}
        \mid X^\ast=a,\widetilde D,\mathcal I_{-k}
    \right]
    -a.
\end{equation*}

\begin{lemma}[Foldwise robust bridge error]
\label{lem:appendix-rs-bridge-error}
Suppose Assumption~\ref{ass:label-conditional-independence} holds, $J$ is fixed, and the true and fitted class contrasts, $\{\eta^\ast_{j,1}(\widetilde D) - \eta^\ast_{j,0}(\widetilde D)\}_{j=1}^J$ and $\{\{\widehat{\eta}^{(-k)}_{j,1}(\widetilde D) - \widehat{\eta}^{(-k)}_{j,0}(\widetilde D)\}_{j=1}^J\}_{k=1}^K$ , are uniformly bounded away from zero. Then, for every fold
$k$, conditional on $\mathcal I_{-k}$,
\begin{align}
    B_{1,-k}^{\mathrm{R}}(\widetilde D)
    ={}&
    -3\binom{J}{2}^{-1}
    \sum_{j_1<j_2}
    e_{j_1,1}^{(-k)}(\widetilde D)
    e_{j_2,1}^{(-k)}(\widetilde D)
    \nonumber\\
    &-
    2\binom{J}{3}^{-1}
    \sum_{j_1<j_2<j_3}
    e_{j_1,1}^{(-k)}(\widetilde D)
    e_{j_2,1}^{(-k)}(\widetilde D)
    e_{j_3,1}^{(-k)}(\widetilde D),
    \label{eq:appendix-rs-error-one}
    \\
    B_{0,-k}^{\mathrm{R}}(\widetilde D)
    ={}&
    3\binom{J}{2}^{-1}
    \sum_{j_1<j_2}
    e_{j_1,0}^{(-k)}(\widetilde D)
    e_{j_2,0}^{(-k)}(\widetilde D)
    \nonumber\\
    &-
    2\binom{J}{3}^{-1}
    \sum_{j_1<j_2<j_3}
    e_{j_1,0}^{(-k)}(\widetilde D)
    e_{j_2,0}^{(-k)}(\widetilde D)
    e_{j_3,0}^{(-k)}(\widetilde D).
    \label{eq:appendix-rs-error-zero}
\end{align}
Moreover, for a finite constant $C$ that does not depend on $n$ or $k$, 
\begin{align}
    \max_{k,j,a}
    \left\|e_{j,a}^{(-k)}\right\|_{2,P}
    &\le C\delta_n,
    \label{eq:appendix-single-bridge-rate}
    \\
    \max_k
    \left\|
        \widehat H_{-k}^{\mathrm{R}}-H_0^{\mathrm{R}}
    \right\|_{2,P}
    &\le C\delta_n, \label{eq:appendix-rs-l2-rate}\\
    \max_{k,a}
    P\left|B_{a,-k}^{\mathrm{R}}\right|
    &\le C\delta_n^2,
    \label{eq:appendix-rs-drift-rate}        
\end{align}
where $\delta_n
    =
    \max_{k,j,a}
    \left\|
        \widehat\eta_{j,a}^{(-k)}-\eta_{j,a}^\ast
    \right\|_{2,P}.$
Consequently, the conditional bridge error has no term that is linear in a
single first-stage error. If all proxy-specific nuisance functions are
correctly specified except possibly those for one proxy, both
$B_{1,-k}^{\mathrm{R}}$ and $B_{0,-k}^{\mathrm{R}}$ are exactly zero.
\end{lemma}

\begin{proof}
Because the fold-specific nuisance estimates are fixed conditional on
$\mathcal I_{-k}$, for $a\in\{0,1\}$,
\begin{align*}
    &\E\!\left[
        \widehat M_{j,-k}
        \mid X^\ast=a,\widetilde D,\mathcal I_{-k}
    \right]
     \ = \ 
    \frac{
        \eta^\ast_{j,a}(\widetilde D)
        -\widehat\eta_{j,0}^{(-k)}(\widetilde D)
    }{
        \widehat\eta_{j,1}^{(-k)}(\widetilde D)-\widehat\eta_{j,0}^{(-k)}(\widetilde D)
    }
\end{align*}
where $\eta^\ast_{j,a}(d) \coloneq \Pr(X^{(j)}=1\mid X^\ast=a,\widetilde D=d).$
Subtracting $a$ yields
\begin{equation*}
    e_{j,a}^{(-k)}(\widetilde D)
    =
    \frac{
        \eta^\ast_{j,a}(\widetilde D)
        -\widehat\eta_{j,a}^{(-k)}(\widetilde D)
    }{
        \widehat\eta_{j,1}^{(-k)}(\widetilde D)-\widehat\eta_{j,0}^{(-k)}(\widetilde D)
    }.
\end{equation*}
Because the absolute value of the denominator is uniformly bounded away from zero,
\begin{equation*}
    \left|e_{j,a}^{(-k)}(\widetilde D)\right|
    \le
    C\left|
        \widehat\eta_{j,a}^{(-k)}(\widetilde D)
        -\eta^\ast_{j,a}(\widetilde D)
    \right|,
\end{equation*}
which proves equation~\eqref{eq:appendix-single-bridge-rate} after taking
$L_2(P)$ norms.

Assumption~\ref{ass:label-conditional-independence} implies that, conditional on
$(X^\ast,\widetilde D,\mathcal I_{-k})$, the fitted single-proxy bridges are mutually independent. When $X^\ast=1$,
\begin{align*}
    &\E\!\left[
        \widehat{\overline H}_{2,-k}
        \mid X^\ast=1,\widetilde D,\mathcal I_{-k}
    \right]
    =
    \binom{J}{2}^{-1}
    \sum_{j_1<j_2}
    \{1+e_{j_1,1}^{(-k)}\}
    \{1+e_{j_2,1}^{(-k)}\},
    \\
    &\E\!\left[
        \widehat{\overline H}_{3,-k}
        \mid X^\ast=1,\widetilde D,\mathcal I_{-k}
    \right]
    =
    \binom{J}{3}^{-1}
    \sum_{j_1<j_2<j_3}
    \{1+e_{j_1,1}^{(-k)}\}
    \{1+e_{j_2,1}^{(-k)}\}
    \{1+e_{j_3,1}^{(-k)}\},
\end{align*}
where the arguments $\widetilde D$ for $e(\cdot)$ are suppressed for readability. Each
$e_{j,1}^{(-k)}$ appears in $J-1$ pairs and $\binom{J-1}{2}$ triples, so the
average linear terms are, respectively,
\begin{equation*}
    \frac{2}{J}\sum_{j=1}^J e_{j,1}^{(-k)}
    \qquad\text{and}\qquad
    \frac{3}{J}\sum_{j=1}^J e_{j,1}^{(-k)}.
\end{equation*}
Their coefficient in
$3\widehat{\overline H}_{2,-k}-2\widehat{\overline H}_{3,-k}$ is zero because
$3(2/J)-2(3/J)=0$. Each pair appears in $J-2$ triples and
\begin{equation*}
    (J-2)
    \binom{J}{3}^{-1}
    =
    3\binom{J}{2}^{-1}.
\end{equation*}
The coefficient on each pair product is therefore
$3\binom{J}{2}^{-1}-2\{3\binom{J}{2}^{-1}\}
=-3\binom{J}{2}^{-1}$, while each triple product has coefficient
$-2\binom{J}{3}^{-1}$. This proves
equation~\eqref{eq:appendix-rs-error-one}.

When $X^\ast=0$, the conditional mean of $\widehat M_{j,-k}$ is
$e_{j,0}^{(-k)}$, so the pair and triple products directly give
equation~\eqref{eq:appendix-rs-error-zero}. By
equation~\eqref{eq:appendix-single-bridge-rate}, Cauchy--Schwarz yields
\begin{equation*}
    P\left|e_{j_1,a}^{(-k)}e_{j_2,a}^{(-k)}\right|
    \le
    \left\|e_{j_1,a}^{(-k)}\right\|_{2,P}
    \left\|e_{j_2,a}^{(-k)}\right\|_{2,P}
    \le
    C\delta_n^2.
\end{equation*}
Because the fitted class contrasts are uniformly bounded away from zero, the single-bridge errors are uniformly bounded, so
each cubic term is bounded by a constant times one of its pair products. Since
$J$ is fixed, summing the finite number of terms proves
equation~\eqref{eq:appendix-rs-drift-rate}.

It remains to prove equation~\eqref{eq:appendix-rs-l2-rate}. Write
\begin{equation*}
    \Delta_j(\widetilde D)
    =
    \eta^\ast_{j,1}(\widetilde D)-\eta^\ast_{j,0}(\widetilde D).
\end{equation*}
The difference between the fitted and true single-label bridges satisfies the
following equality.
\begin{align*}
    \widehat M_{j,-k}-M_{j,0}
     \ = \ 
    \frac{
        \eta^\ast_{j,0}-\widehat\eta_{j,0}^{(-k)}
    }{
        \widehat\eta_{j,1}^{(-k)}-\widehat\eta_{j,0}^{(-k)}
    }
    \ + \
    \{X^{(j)}-\eta^\ast_{j,0}\}
    \frac{
        \bigl(\eta^\ast_{j,1}-\widehat\eta_{j,1}^{(-k)}\bigr)
        -
        \bigl(\eta^\ast_{j,0}-\widehat\eta_{j,0}^{(-k)}\bigr)
    }{
        \bigl(\widehat\eta_{j,1}^{(-k)}-\widehat\eta_{j,0}^{(-k)}\bigr)
        \bigl(\eta^\ast_{j,1}-\eta^\ast_{j,0}\bigr)
    },
\end{align*}
where all nuisance functions are evaluated at $\widetilde D$. Because $X^{(j)}$ and the proxy means lie in $[0,1]$ and both contrasts are bounded away from zero,
\begin{equation*}
    \left|\widehat M_{j,-k}-M_{j,0}\right|
    \le
    C\max_{a\in\{0,1\}}
    \left|
        \widehat\eta_{j,a}^{(-k)}-\eta^\ast_{j,a}
    \right|.
\end{equation*}
Hence
\begin{equation*}
    \max_{k,j}
    \left\|\widehat M_{j,-k}-M_{j,0}\right\|_{2,P}
    \le
    C\delta_n.
\end{equation*}
For any $S=\{j_1,\ldots,j_m\}$ with $m\in\{2,3\}$, inserting one fitted
factor at a time gives
\begin{equation*}
    \prod_{t=1}^m\widehat M_{j_t,-k}
    -\prod_{t=1}^m M_{j_t,0}
    =
    \sum_{t=1}^m
    \left(\prod_{s<t}\widehat M_{j_s,-k}\right)
    (\widehat M_{j_t,-k}-M_{j_t,0})
    \left(\prod_{s>t}M_{j_s,0}\right).
\end{equation*}
Because the true and fitted class contrasts are uniformly bounded away from zero, all factors multiplying each difference are uniformly bounded. 
Taking $L_2(P)$ norms, applying the triangle inequality,
and averaging the finitely many pair and triple products proves
equation~\eqref{eq:appendix-rs-l2-rate}.
\end{proof}

\begin{proof}[Proof of Proposition~\ref{prop:robust-symmetric-bridge}]
The conditional error identities follow from
Lemma~\ref{lem:appendix-rs-bridge-error}. Applying the lemma at the true
nuisance functions gives $e_{j,a}^{(-k)}=0$ for every $j$ and
$a\in\{0,1\}$, yielding the unbiasedness of the bridge function. The displayed identities
contain only products involving errors from at least two distinct labels and
therefore contain no term that is linear in a single first-stage error.
Moreover, if all proxy-specific nuisance functions are correctly specified
except possibly those associated with one proxy, every pair and triple product
of the errors vanishes, yielding the stated multi-proxy robustness.
\end{proof}

\subsection{Proof of Theorem~\ref{thm:latent-independent-asymptotics}: Second-Order Remainder for the DMM Moment Function and Large-Sample Theory}
\label{app:proof-rs-asymptotics}
Let $\mathcal B_0\subset\mathcal B$ be the compact neighborhood of
$\beta^\ast$ in Theorem~\ref{thm:latent-independent-asymptotics}.
Under the bounded-envelope condition in
Theorem~\ref{thm:latent-independent-asymptotics}, there exists a finite
constant $C$ such that
\begin{equation*}
    \sup_{\beta\in\mathcal B_0}
    \left\|\psi^F(Y,1,W;\beta)
    -
    \psi^F(Y,0,W;\beta)\right\|
    \le C
    \qquad\text{a.s.}
\end{equation*}
The remaining smoothness and envelope conditions are those stated in
Theorem~\ref{thm:latent-independent-asymptotics}.

For each fold, define the conditional population remainder
\begin{equation*}
    R_{n,k}(\beta)
    =
    \E\!\left[
        \psi^{\mathrm{DMM}}(O;\beta,\widehat\eta_{-k})
        -\psi^{\mathrm{DMM}}(O;\beta,\eta^\ast)
        \mid \{O_i:i\in\mathcal I_{-k}\}
    \right]
\end{equation*}
where $O = (Y, \widetilde{X}, \widetilde{D})$, and its fold-weighted average
\begin{equation*}
    \overline R_n(\beta)
    =
    \sum_{k=1}^K\frac{n_k}{n}R_{n,k}(\beta).
\end{equation*}

\begin{lemma}[Second-order remainder of the DMM moment function]
\label{lem:appendix-rs-score-drift}
Under the conditions of
Lemma~\ref{lem:appendix-rs-bridge-error} and the envelope conditions above,
for every fold $k$ and $\beta\in\mathcal B_0$,
\begin{align}
    R_{n,k}(\beta)
    =
    \E\Big[{}
    \{X^\ast B_{1,-k}^{\mathrm{R}}(\widetilde D)
      +(1-X^\ast)B_{0,-k}^{\mathrm{R}}(\widetilde D)\}
    \ \times \ \Delta\psi(Y,W;\beta)
    \;\Bigm|\;\mathcal I_{-k}
    \Big].
    \label{eq:appendix-drift-identity}
\end{align}
Consequently,
\begin{equation}
    \max_{1\le k\le K}
    \sup_{\beta\in\mathcal B_0}
    \left\|R_{n,k}(\beta)\right\|_2
    \le
    C\delta_n^2,
    \label{eq:appendix-quadratic-drift}
\end{equation}
and the same bound holds with $R_{n,k}$ replaced by $\overline R_n$.
\end{lemma}

\begin{proof}
The difference between the fitted and population estimating functions has the
exact form
\begin{align*}
    &\psi^{\mathrm{DMM}}(O;\beta,\widehat\eta_{-k})
    -\psi^{\mathrm{DMM}}(O;\beta,\eta^\ast)
     \ = \ 
    (\widehat H_{-k}^{\mathrm{R}}-H_0^{\mathrm{R}})
    \Delta\psi(Y,W;\beta)
\end{align*}
where
$\widehat H_{-k}^{\mathrm{R}}
\coloneq
H^{\mathrm{R}}(\widetilde{X}, \widetilde{D};\widehat\eta_{-k})
=
3\binom{J}{2}^{-1}
\sum_{j_1<j_2}
\widehat M_{j_1,-k}\widehat M_{j_2,-k}
-
2\binom{J}{3}^{-1}
\sum_{j_1<j_2<j_3}
\widehat M_{j_1,-k}\widehat M_{j_2,-k}\widehat M_{j_3,-k},$
and
$H_0^{\mathrm{R}}
\coloneq
H^{\mathrm{R}}(\widetilde{X}, \widetilde{D};\eta^\ast)$ is defined analogously by replacing
$\widehat M_{j,-k}$ with $M_{j,0}$. These are the robust bridges evaluated at
the fold-specific estimated and true nuisance functions, respectively.
Conditioning on $(X^\ast,\widetilde D,\mathcal I_{-k})$ and using
$\E[H_0^{\mathrm{R}}\mid X^\ast,\widetilde D]=X^\ast$ gives
\begin{equation*}
    \E\!\left[
        \widehat H_{-k}^{\mathrm{R}}-H_0^{\mathrm{R}}
        \mid X^\ast,\widetilde D,\mathcal I_{-k}
    \right]
    =
    X^\ast B_{1,-k}^{\mathrm{R}}(\widetilde D)
    +(1-X^\ast)B_{0,-k}^{\mathrm{R}}(\widetilde D).
\end{equation*}
Substitution proves equation~\eqref{eq:appendix-drift-identity}. By that
identity and the bounded moment-contrast condition,
\begin{align*}
    \sup_{\beta\in\mathcal B_0}
    \left\|R_{n,k}(\beta)\right\|_2
    &\le
    \E\!\left[
        \left\{
            \left|B_{1,-k}^{\mathrm R}(\widetilde D)\right|
            +
            \left|B_{0,-k}^{\mathrm R}(\widetilde D)\right|
        \right\}
        \sup_{\beta\in\mathcal B_0}
        \left\|\Delta\psi(Y,W;\beta)\right\|
        \,\middle|\,\mathcal I_{-k}
    \right]
    \\
    &\le
    C\sum_{a=0}^1
    P\left|B_{a,-k}^{\mathrm R}\right|
    \le
    C\delta_n^2,
\end{align*}
where the last inequality follows from
equation~\eqref{eq:appendix-rs-drift-rate}. This proves
equation~\eqref{eq:appendix-quadratic-drift}. Because $K$ is fixed and the
weights $n_k/n$ sum to one, the same rate holds for $\overline R_n$.
\end{proof}

\begin{proof}[Proof of Theorem~\ref{thm:latent-independent-asymptotics}]
Define the cross-fitted sample moment and its population counterpart by
\begin{equation*}
    \widehat\Psi_n(\beta)
    =
    \sum_{k=1}^K\frac{n_k}{n}
    \mathbb P_{n,k}
    \psi^{\mathrm{DMM}}(O;\beta,\widehat\eta_{-k}),
    \qquad
    \Psi(\beta)
    =
    P\psi^{\mathrm{DMM}}(O;\beta,\eta^\ast).
\end{equation*}
By Theorem~\ref{thm:latent-independent-identification},
$\Psi(\beta^\ast)=0$ and $\beta^\ast$ is its unique root. For every
$\beta\in\mathcal B_0$,
\begin{align}
    \widehat\Psi_n(\beta)-\Psi(\beta)
    ={}&
    \sum_{k=1}^K\frac{n_k}{n}
    (\mathbb P_{n,k}-P)
    \Big\{
        \psi^{\mathrm{DMM}}(O;\beta,\widehat\eta_{-k})
        -\psi^{\mathrm{DMM}}(O;\beta,\eta^\ast)
    \Big\}
    \nonumber\\
    &+
    (\mathbb P_n-P)\psi^{\mathrm{DMM}}(O;\beta,\eta^\ast)
    +\overline R_n(\beta).
    \label{eq:appendix-moment-decomposition}
\end{align}

\paragraph{Consistency.} We first establish consistency. We begin with the first term in equation~\eqref{eq:appendix-moment-decomposition}.
Conditional on $\mathcal I_{-k}$, the
observations in $\mathcal I_k$ are independent of $\widehat\eta_{-k}$. By the
exact moment-difference identity, the bounded moment-contrast condition, and
equation~\eqref{eq:appendix-rs-l2-rate},
\begin{align*}
    \max_k
    \sup_{\beta\in\mathcal B_0}
    \left\|
        \psi^{\mathrm{DMM}}(O;\beta,\widehat\eta_{-k})
        -\psi^{\mathrm{DMM}}(O;\beta,\eta^\ast)
    \right\|_{2,P}
    &=
    \max_k
    \sup_{\beta\in\mathcal B_0}
    \left\|
        (\widehat H_{-k}^{\mathrm R}-H_0^{\mathrm R})
        \Delta\psi(Y,W;\beta)
    \right\|_{2,P}
    \\
    &\le
    C\max_k
    \left\|\widehat H_{-k}^{\mathrm R}-H_0^{\mathrm R}\right\|_{2,P}
    \le C\delta_n
    =o_p(1).
\end{align*}

For any $r>0$, let $\{\beta_1,\ldots,\beta_N\}$ be a finite
$r$-net of $\mathcal B_0$. For each fold $k$ and $\epsilon>0$,
conditional Chebyshev's inequality and a union bound give
\begin{align*}
    &\Pr\Bigg[
        \max_{1\le \ell\le N}
        \left\|
            (\mathbb P_{n,k}-P)
            \Big\{
                \psi^{\mathrm{DMM}}
                (O;\beta_\ell,\widehat\eta_{-k})
                -
                \psi^{\mathrm{DMM}}
                (O;\beta_\ell,\eta^\ast)
            \Big\}
        \right\|_2
        >
        \epsilon
        \,\Bigg|\,
        \mathcal I_{-k}
    \Bigg]
    \\
    &\qquad\le
    \frac{1}{n_k\epsilon^2}
    \sum_{\ell=1}^N
    \left\|
        \psi^{\mathrm{DMM}}
        (O;\beta_\ell,\widehat\eta_{-k})
        -
        \psi^{\mathrm{DMM}}
        (O;\beta_\ell,\eta^\ast)
    \right\|_{2,P}^2
    =
    o_p(1).
\end{align*}
Because $K$ is fixed, this also holds uniformly over $k$. Moreover,
class-contrast separation uniformly bounds the bridge difference, and the
Lipschitz-envelope condition for the moment contrast implies
\begin{align*}
    &\max_k
    \sup_{\beta\in\mathcal B_0}
    \left\|
        (\mathbb P_{n,k}-P)
        \Big\{
            \psi^{\mathrm{DMM}}(O;\beta,\widehat\eta_{-k})
            -
            \psi^{\mathrm{DMM}}(O;\beta,\eta^\ast)
        \Big\}
    \right\|_2
    \\
    &\qquad\le
    \max_{\substack{1\le k\le K\\1\le\ell\le N}}
    \left\|
        (\mathbb P_{n,k}-P)
        \Big\{
            \psi^{\mathrm{DMM}}(O;\beta_\ell,\widehat\eta_{-k})
            -
            \psi^{\mathrm{DMM}}(O;\beta_\ell,\eta^\ast)
        \Big\}
    \right\|_2
    +
    O_p(r)
    \\
    &\qquad=
    o_p(1)+O_p(r).
\end{align*}
Letting $n\to\infty$ and then $r\downarrow0$ yields
\begin{equation*}
    \max_k
    \sup_{\beta\in\mathcal B_0}
    \left\|
        (\mathbb P_{n,k}-P)
        \Big\{
            \psi^{\mathrm{DMM}}(O;\beta,\widehat\eta_{-k})
            -
            \psi^{\mathrm{DMM}}(O;\beta,\eta^\ast)
        \Big\}
    \right\|_2
    =
    o_p(1).
\end{equation*}
Since the weights $n_k/n$ sum to one, the first term in
equation~\eqref{eq:appendix-moment-decomposition} is therefore
uniformly $o_p(1)$.

Now, we consider the second term in equation~\eqref{eq:appendix-moment-decomposition}. The oracle moment function class is finite-dimensional and has an
integrable Lipschitz envelope, so a uniform law of large numbers makes the
second term uniformly $o_p(1)$. 

Finally, we consider the third term in equation~\eqref{eq:appendix-moment-decomposition}. 
Lemma~\ref{lem:appendix-rs-score-drift} and $\delta_n=o_p(1)$ make the third
term uniformly $o_p(1)$. Therefore,
\begin{equation*}
    \sup_{\beta\in\mathcal B_0}
    \left\|
        \widehat\Psi_n(\beta)-\Psi(\beta)
    \right\|_2
    =o_p(1).
\end{equation*}
Because $\beta^\ast$ is a well-separated root and
$\widehat\Psi_n(\widehat\beta^{\mathrm{DMM}})=0$, the standard
$Z$-estimation argument yields
$\widehat\beta^{\mathrm{DMM}}\overset{p}{\to}\beta^\ast$.

The robustness claim follows similarly. If the nuisance functions for all
but one proxy are consistently estimated, every pair or triple in the
conditional bridge-error formulas contains at least one $o_p(1)$ factor,
while the other factors remain bounded under the regularity
condition. Hence, the remainder term is uniformly $o_p(1)$. For the centered
empirical term, equation~\eqref{eq:appendix-rs-l2-rate} need not hold when one
proxy-specific nuisance limit is misspecified. Nevertheless, class-contrast
separation uniformly bounds both the fitted and true bridges, and the
downstream moment contrast is uniformly bounded. Conditional on the training
sample, the corresponding centered foldwise average is therefore
$O_p(n^{-1/2})$, and hence $o_p(1)$. The same $Z$-estimation argument yields
consistency even if one proxy-specific nuisance model is misspecified.

\paragraph{Asymptotic Normality.} 
We next establish asymptotic normality. At $\beta^\ast$,
\begin{equation}
    \widehat\Psi_n(\beta^\ast)
    \ = \
    (\mathbb P_n-P)
    \psi^{\mathrm{DMM}}(O;\beta^\ast,\eta^\ast)
    +
    \sum_{k=1}^K\frac{n_k}{n}
    (\mathbb P_{n,k}-P)
    \Big\{
        \psi^{\mathrm{DMM}}(O;\beta^\ast,\widehat\eta_{-k})
        -
        \psi^{\mathrm{DMM}}(O;\beta^\ast,\eta^\ast)
    \Big\}
    +
    \overline R_n(\beta^\ast)
    \label{eq:appendix-moment-at-target}
\end{equation}
because $P\psi^{\mathrm{DMM}}(O;\beta^\ast,\eta^\ast)=0$.

Conditional on the training observations $\{O_i:i\in\mathcal I_{-k}\}$,
\begin{align*}
    &\E\!\left[
        \left\|
            \sqrt{n_k}(\mathbb P_{n,k}-P)
            \left\{
                \psi^{\mathrm{DMM}}(O;\beta^\ast,\widehat\eta_{-k})
                -\psi^{\mathrm{DMM}}(O;\beta^\ast,\eta^\ast)
            \right\}        
        \right\|_2^2
        \;\middle|\;
        \{O_i:i\in\mathcal I_{-k}\}
    \right]\\
    &\qquad\le
    C\left\|
        \psi^{\mathrm{DMM}}(O;\beta^\ast,\widehat\eta_{-k})
        -\psi^{\mathrm{DMM}}(O;\beta^\ast,\eta^\ast)
    \right\|_{2,P}^2
    \\
    &\qquad\le
    C\left\|
        \widehat H_{-k}^{\mathrm{R}}-H_0^{\mathrm{R}}
    \right\|_{2,P}^2
    \le C\delta_n^2
    =o_p(1),
\end{align*}
where the second inequality uses the exact moment-difference identity and the
bounded moment-contrast condition, and the third inequality follows from
equation~\eqref{eq:appendix-rs-l2-rate}. Conditional Markov's inequality and
fixed $K$ therefore
imply that the second term in
equation~\eqref{eq:appendix-moment-at-target} is $o_p(n^{-1/2})$.
Lemma~\ref{lem:appendix-rs-score-drift} and
$\delta_n^2=o_p(n^{-1/2})$ imply that the third term is also
$o_p(n^{-1/2})$. Hence
\begin{equation}
    \widehat\Psi_n(\beta^\ast)
    =
    (\mathbb P_n-P)\psi^{\mathrm{DMM}}(O;\beta^\ast,\eta^\ast)
    +o_p(n^{-1/2}).
    \label{eq:appendix-asymptotic-moment}
\end{equation}

The integral form of the mean-value expansion gives
\begin{align*}
    0
    ={}&
    \widehat\Psi_n(\beta^\ast)
    \\
    &+
    \Bigg[
    \int_0^1
    \sum_{k=1}^K\frac{n_k}{n}
    \mathbb P_{n,k}\Bigg\{
        \{1-\widehat H_{-k}^{\mathrm{R}}\}
        \frac{\partial}{\partial\beta^\top}
        \psi^F\!\left(
            Y,0,W;
            \beta^\ast+t(\widehat\beta^{\mathrm{DMM}}-\beta^\ast)
        \right)
    \\
    &\hspace{5.5cm}
        +
        \widehat H_{-k}^{\mathrm{R}}
        \frac{\partial}{\partial\beta^\top}
        \psi^F\!\left(
            Y,1,W;
            \beta^\ast+t(\widehat\beta^{\mathrm{DMM}}-\beta^\ast)
        \right)
    \Bigg\}
    \,dt
    \Bigg]
    (\widehat\beta^{\mathrm{DMM}}-\beta^\ast).
\end{align*}
Consistency, the uniform law of large numbers, smoothness of the
moment functions, and equation~\eqref{eq:appendix-rs-l2-rate} imply that the
matrix in brackets converges in probability to $-\mathbf A_0$.
Since $\mathbf A_0$ is nonsingular, combining this result with
equation~\eqref{eq:appendix-asymptotic-moment} yields
\begin{equation*}
    \sqrt n(\widehat\beta^{\mathrm{DMM}}-\beta^\ast)
    =
    \mathbf{A}_0^{-1}
    \frac{1}{\sqrt n}
    \sum_{i=1}^n
    \psi^{\mathrm{DMM}}(O_i;\beta^\ast,\eta^\ast)
    +o_p(1),
\end{equation*}
which is equation~\eqref{eq:rs-asymptotic-linearity}. The multivariate central
limit theorem gives
\begin{equation*}
    \frac{1}{\sqrt n}
    \sum_{i=1}^n
    \psi^{\mathrm{DMM}}(O_i;\beta^\ast,\eta^\ast)
    \rightsquigarrow
    N(0,\boldsymbol{\Omega}_0),
\end{equation*}
and Slutsky's theorem proves
equation~\eqref{eq:rs-asymptotic-normality}.

Finally, consistency and equation~\eqref{eq:appendix-rs-l2-rate} imply that
the cross-fitted moment function evaluated at
$\widehat\beta^{\mathrm{DMM}}$ converges in $L_2(P)$ to
$\psi^{\mathrm{DMM}}(O;\beta^\ast,\eta^\ast)$. Foldwise laws of large numbers then give
$\widehat{\boldsymbol{\Omega}}\overset{p}{\to}\boldsymbol{\Omega}_0$ and
$\widehat{\mathbf{A}}\overset{p}{\to}\mathbf{A}_0$. Therefore,
\begin{equation*}
    \widehat{\mathbf{V}}
    \overset{p}{\to}
    \mathbf{A}_0^{-1}\boldsymbol{\Omega}_0\mathbf{A}_0^{-\top}.
\end{equation*}
\end{proof}

\section{Details for the Efficiency Discussion}
\label{app:rs-efficiency}

Throughout this section, write $H_{0,J}^{\mathrm{R}}$ when emphasizing
that the robust bridge uses $J$ proxies.
Equation~\eqref{eq:variance-rs-score-decomposition} gives
$\psi_J^{\mathrm{DMM}}(\widetilde X,\widetilde D;\beta^\ast,\eta^\ast)
=
\psi_0^F
+\{H_{0,J}^{\mathrm{R}}-X^\ast\}\Delta\psi_0, 
$
where $\psi_0^F$ is a shorthand for the oracle full-data moment function evaluated at the true parameter.
The unbiasedness of the bridge function
$\E\left[
H_{0,J}^{\mathrm{R}}
\mid X^\ast,\widetilde D
\right]
=
X^\ast
$
implies equation~\eqref{eq:variance-oracle-bread}. Expanding the product
defining the meat matrix yields
\begin{align*}
    \boldsymbol\Omega_{0,J}^{\mathrm{R}}
    ={}&
    \E[\psi_0^F\psi_0^{F\top}]
    +
    \E\left[
        (H_{0,J}^{\mathrm{R}}-X^\ast)
        \psi_0^F\Delta\psi_0^\top
    \right]
    \\
    &+
    \E\left[
        (H_{0,J}^{\mathrm{R}}-X^\ast)
        \Delta\psi_0\psi_0^{F\top}
    \right]
    \\
    &+
    \E\left[
        (H_{0,J}^{\mathrm{R}}-X^\ast)^2
        \Delta\psi_0\Delta\psi_0^\top
    \right].
\end{align*}
The two cross terms vanish after conditioning on
$(X^\ast,\widetilde D)$, which proves
equation~\eqref{eq:variance-meat-decomposition}.

We may write
\begin{align*}
    \E\left[
        (H_{0,J}^{\mathrm{R}}-X^\ast)^2
        \Delta\psi_0\Delta\psi_0^\top
    \right]
    &=
    \E\left[
        \E[(H_{0,J}^{\mathrm{R}}-X^\ast)^2
        \mid X^\ast, \widetilde D]
        \Delta\psi_0\Delta\psi_0^\top
    \right] \\
    &=
    \E\left[
        \sum_{a=0}^{1} \Pr(X^\ast = a \mid \widetilde D)\E[(H_{0,J}^{\mathrm{R}}-a)^2\mid X^\ast = a, \widetilde D]
        \Delta\psi_0\Delta\psi_0^\top
    \right] \\
    &=
    \E\left[
        \sum_{a=0}^1
        \pi(\widetilde D_i)^a
        \{1-\pi(\widetilde D_i)\}^{1-a}
        v_{a,J}^{\mathrm{R}}(\widetilde D)
        \Delta\psi_0\Delta\psi_0^\top
    \right].
\end{align*}
Premultiplying and postmultiplying by $\mathbf A_0^{-1}$ and
$\mathbf A_0^{-\top}c$, respectively, gives
equation~\eqref{eq:variance-coefficient-factorization}.

Define the centered single-proxy bridge noise
\begin{equation*}
    \varepsilon_j
    =
    M_j(\widetilde D; \eta)-X^\ast.
\end{equation*}
Conditional on
$(X^\ast,\widetilde D)=(a,d)$, the variables
$\varepsilon_1,\ldots,\varepsilon_J$ are mutually independent and satisfy
$\E[\varepsilon_j\mid X^\ast=a,\widetilde D=d]=0$,
$\E[\varepsilon_j^2\mid X^\ast=a,\widetilde D=d]=\kappa_{j,a}(d)$.
Direct expansion of the pair and triple averages gives
\begin{equation}
    H_{0,J}^{\mathrm{R}}-X^\ast
    =
    \frac{3(1-2X^\ast)}{\binom{J}{2}}
    \sum_{j_1<j_2}\varepsilon_{j_1}\varepsilon_{j_2}
    -
    \frac{2}{\binom{J}{3}}
    \sum_{j_1<j_2<j_3}\varepsilon_{j_1}\varepsilon_{j_2}\varepsilon_{j_3}.
    \label{eq:appendix-rs-centered}
\end{equation}
For any nonempty subset $S\subseteq\{1,\ldots,J\}$, let
$\varepsilon_S=\prod_{j\in S}\varepsilon_j$. If $S\neq T$, then
\begin{align*}
    &\E\left[
        \varepsilon_S\varepsilon_T
        \mid X^\ast=a,\widetilde D=d
    \right]
    \\
    &\qquad=
    \prod_{j\in S\cap T}
    \E\left[
        \varepsilon_j^2
        \mid X^\ast=a,\widetilde D=d
    \right]
    \prod_{j\in (S \setminus T) \cup (T \setminus S)}
    \E\left[
        \varepsilon_j
        \mid X^\ast=a,\widetilde D=d
    \right]
    =0,
\end{align*}
because $(S \setminus T) \cup (T \setminus S)$ is nonempty. Thus, distinct pair and triple monomials
are conditionally orthogonal: all pair--pair, triple--triple, and pair--triple
covariance terms vanish. For $S=T$,
$\E\left[\varepsilon_S^2\mid X^\ast=a,\widetilde D=d\right]
=\prod_{j\in S}\kappa_{j,a}(d)$.
Squaring equation~\eqref{eq:appendix-rs-centered} and taking conditional
expectations therefore yields
\begin{equation*}
    v_{a,J}^{\mathrm{R}}(d)
    =
    9
    \binom{J}{2}^{-2}
    \sum_{j_1<j_2}
    \kappa_{j_1,a}(d)\kappa_{j_2,a}(d)
    +
    4
    \binom{J}{3}^{-2}
    \sum_{j_1<j_2<j_3}
    \kappa_{j_1,a}(d)\kappa_{j_2,a}(d)\kappa_{j_3,a}(d),
\end{equation*}
which is equation~\eqref{eq:variance-rs-exact}.

Our main question is: when does an additional label improve precision?
Define the elementary symmetric sums
\begin{align*}
    S_{1,a,J}(d)
    &=
    \sum_{j_1=1}^J \kappa_{j_1,a}(d),
    \\
    S_{2,a,J}(d)
    &=
    \sum_{j_1<j_2}
    \kappa_{j_1,a}(d)\kappa_{j_2,a}(d),
    \\
    S_{3,a,J}(d)
    &=
    \sum_{j_1<j_2<j_3}
    \kappa_{j_1,a}(d)\kappa_{j_2,a}(d)\kappa_{j_3,a}(d).
\end{align*}
Suppose a new label has conditional bridge variance
$\kappa_{J+1,a}(d)$. The elementary symmetric sums update according to
$S_{2,a,J+1}(d)=S_{2,a,J}(d)+\kappa_{J+1,a}(d)S_{1,a,J}(d),$
$S_{3,a,J+1}(d)=S_{3,a,J}(d)+\kappa_{J+1,a}(d)S_{2,a,J}(d).$
Consequently,
\begin{align}
    &v_{a,J+1}^{\mathrm{R}}(d)
    -v_{a,J}^{\mathrm{R}}(d)
    \nonumber\\
    &\quad=
    \kappa_{J+1,a}(d)
    \left\{
        9S_{1,a,J}(d)
        \binom{J+1}{2}^{-2}
        +
        4S_{2,a,J}(d)
        \binom{J+1}{3}^{-2}
    \right\}
    \nonumber\\
    &\qquad-
    \left[
        9S_{2,a,J}(d)
        \left\{
            \binom{J}{2}^{-2}-\binom{J+1}{2}^{-2}
        \right\}
        +
        4S_{3,a,J}(d)
        \left\{
            \binom{J}{3}^{-2}-\binom{J+1}{3}^{-2}
        \right\}
    \right].
    \label{eq:appendix-rs-variance-difference}
\end{align}
If the denominator below is positive, define
\begin{equation}
    \overline{\kappa}_{J,a}(d)
    \coloneq
    \frac{
        9S_{2,a,J}(d)
        \{\binom{J}{2}^{-2}-\binom{J+1}{2}^{-2}\}
        +
        4S_{3,a,J}(d)
        \{\binom{J}{3}^{-2}-\binom{J+1}{3}^{-2}\}
    }{
        9S_{1,a,J}(d)\binom{J+1}{2}^{-2}
        +
        4S_{2,a,J}(d)\binom{J+1}{3}^{-2}
    }.
    \label{eq:appendix-rs-threshold}
\end{equation}
Equation~\eqref{eq:appendix-rs-variance-difference} then implies
$v_{a,J+1}^{\mathrm{R}}(d)
<
v_{a,J}^{\mathrm{R}}(d)
\Longleftrightarrow
\kappa_{J+1,a}(d)<\overline{\kappa}_{J,a}(d)$.
If this denominator is zero, then all current $\kappa_{j,a}(d)$ are zero, and both
conditional bridge variances remain zero after adding a single additional
label.

For a particular coefficient contrast, subtracting
equation~\eqref{eq:variance-coefficient-factorization} for $J$ and $J+1$
gives the exact condition
\begin{align*}
    &c^\top
    \left(
        \mathbf V_{0,J+1}^{\mathrm{R}}
        -\mathbf V_{0,J}^{\mathrm{R}}
    \right)c
    \\
    &\quad=
    \E\left[
        \sum_{a=0}^1
        \pi(\widetilde D_i)^a
        \{1-\pi(\widetilde D_i)\}^{1-a}
        \{v_{a,J+1}^{\mathrm{R}}(\widetilde D)
          -v_{a,J}^{\mathrm{R}}(\widetilde D)\}
        \{(\mathbf{A}_0^{-\top}c)^\top\Delta\psi_0(Y,W)\}^2
    \right].
\end{align*}
Therefore, satisfying the threshold for both latent classes almost surely is a
simple sufficient condition for weakly improving every coefficient contrast.
For a fixed contrast, it is sufficient and necessary that the displayed
weighted average be negative.

\paragraph{A special case.}
If the existing proxies have common conditional bridge variance
$\kappa_a(d)>0$ and the added proxy has conditional bridge variance $q_a(d)$,
equation~\eqref{eq:appendix-rs-threshold} reduces to
\begin{equation}
    q_a(d)
    <
    2\kappa_a(d)
    \frac{
        J(J-2)+(2J-1)\kappa_a(d)
    }{
        (J-2)\{J-1+2\kappa_a(d)\}
    }.
    \label{eq:appendix-rs-homogeneous-threshold}
\end{equation}
In particular, an additional proxy of the same quality,
$q_a(d)=\kappa_a(d)$, always satisfies this condition. If all $J$ proxies have
common variance $\kappa_a(d)$, then
$v_{a,J}^{\mathrm{R}}(d)
=
\frac{18\kappa_a(d)^2}{J(J-1)}
+
\frac{24\kappa_a(d)^3}{J(J-1)(J-2)},$
which is strictly decreasing in $J$ for $\kappa_a(d)>0$ and is of order
$J^{-2}$. Moreover, the right-hand side of
equation~\eqref{eq:appendix-rs-homogeneous-threshold} converges to
$2\kappa_a(d)$ as $J\to\infty$. Thus, in this case, an added proxy may be
somewhat noisier than the existing labels and still improve precision through
the additional averaging.

All formulas in this subsection concern the population bridge evaluated at the
true nuisance functions. If conditional independence across proxies fails,
additional covariance terms generally enter the conditional variance, and the
threshold above no longer applies. These comparisons also hold across designs
with different fixed values of $J$; a formal analysis with $J=J_n\to\infty$
would require additional control of the growing nuisance dimension and the minimum class contrast.

\section{Additional Details on the Experiments}
\label{app:application-details}

This appendix provides additional information on the two studies in
Section~\ref{sec:applications}. We first clarify how our analyses relate to the
original applications, then describe the logistic latent-class model and the
expectation--maximization (EM) algorithm used in both studies, and finally give
application-specific details. Throughout, an ``expert label'' denotes the
benchmark annotation.

\subsection{Scope and Relationship to the Original Applications}
\label{app:applications-scope}

The two analyses are numerical illustrations of DMM rather than exact
replications of the original substantive studies. In the first application,
the original analysis and the DSL validation study examine differences in
political-advertising tone between Facebook and television using the original
expert-coded corpus and a candidate-adjusted downstream analysis
\citep{fowler2021political,egami2024using}. We instead use a transparent
ad-level logistic regression for a binary Promote indicator. The original
annotation file contains $13{,}040$ expert-coded ads. After merging it with the
platform and candidate covariates and retaining complete records, our analysis
contains $12{,}973$ unique ads. 
Importantly, our Facebook coefficient is not the same
estimand as the candidate-fixed-effects coefficient in the original study.

In the second application, the original study codes both prefecture-level
and county-level wrongdoing and includes both variables in a richer downstream
specification \citep{pan2018concealing}. We focus on prefecture-level
wrongdoing because it enters the simplified downstream model as a single main
effect. This isolates one latent independent variable and avoids introducing a
second latent variable and the county-wrongdoing interaction used in the
original analysis. The resulting coefficient should therefore be interpreted
as the coefficient from the working model defined below, not as an exact
replication of every specification in \citet{pan2018concealing}.

Expert labels serve three different roles across the exercises. First, they
define the full-sample benchmark in both applications. Second, they are used to
calibrate the synthetic Fowler data-generating process. Third, they are used
retrospectively to rank the available proxy labels by F1 score. They do not
enter the feasible DMM moment equations within a Monte Carlo
replication or the DMM fit. The F1-based ranking should consequently
be viewed as a device for constructing interpretable nested proxy sets
for evaluation.

\subsection{Model Specification}
\label{app:shared-logistic-em}

Both applications use the same low-dimensional logistic latent-class model for
the measurement nuisance functions. To describe the common implementation,
let $L_i^\ast\in\{0,1\}$ denote a generic latent binary label, let
$\widetilde L_i=(L_i^{(1)},\ldots,L_i^{(J)})$ collect its binary proxy labels,
and let $Z_i$ denote the nuisance-model design vector, including an intercept.
The correspondence is
\begin{equation*}
\begin{array}{lll}
\text{Fowler:}
& L_i^\ast=Y_i^\ast,
& Z_i=(1,T_i,W_i^\top)^\top,\\[2pt]
\text{Pan--Chen:}
& L_i^\ast=X_i^\ast,
& Z_i=(1,W_i^\top)^\top
\end{array}
\end{equation*}
where $T_i$ is the main independent variable, an indicator for Facebook advertisements.
We posit
\begin{align*}
    p_i
    &\coloneq{\Pr}(L_i^\ast=1\mid Z_i)
      =\operatorname{expit}(Z_i^\top\alpha),\\
    \eta_{ij,a}
    &\coloneq \eta_{j,a}(Z_i)
    \coloneq{\Pr}(L_i^{(j)}=1\mid L_i^\ast=a,Z_i)
      =\operatorname{expit}(Z_i^\top\gamma_{j,a}),
      \qquad a\in\{0,1\},
\end{align*}
where $\operatorname{expit}(u)=\{1+\exp(-u)\}^{-1}$. Conditional independence implies the
observed-data likelihood
\begin{equation}
    \mathcal L(\theta)
    =\prod_{i=1}^n\Bigg[{}
    (1-p_i)
    \prod_{j=1}^J
        \eta_{ij,0}^{L_i^{(j)}}
        (1-\eta_{ij,0})^{1-L_i^{(j)}}\\
    +
    p_i
    \prod_{j=1}^J
        \eta_{ij,1}^{L_i^{(j)}}
        (1-\eta_{ij,1})^{1-L_i^{(j)}}
    \Bigg],
    \label{eq:app-em-observed-likelihood}
\end{equation}
where $\theta=(\alpha,\{\gamma_{j,a}\}_{j,a})$. The downstream dependent
variable is not included in this likelihood. 
In particular, in the Pan--Chen
analysis, upward reporting is not used to infer the latent wrongdoing label.
Thus, we impose two substantively reasonable conditional independence assumptions and use the corresponding nuisance models: $X^{(1)} \indep \cdots \indep X^{(J)} \mid X^\ast, W$ and $\widetilde X \indep Y \mid X^\ast, W$, which jointly imply Assumption~\ref{ass:label-conditional-independence}.

\subsubsection{Expectation--Maximization Algorithm}
\label{app:em-algorithm}

Let
$\xi_i^{(t)}={\Pr}_{\theta^{(t)}}(L_i^\ast=1\mid\widetilde L_i,Z_i)$
denote the posterior responsibility at iteration $t$. For compactness, define
\begin{equation*}
    f_{ij,a}^{(t)}
    =
    \{\eta_{ij,a}^{(t)}\}^{L_i^{(j)}}
    \{1-\eta_{ij,a}^{(t)}\}^{1-L_i^{(j)}}.
\end{equation*}
The E-step is
\begin{equation}
    \xi_i^{(t)}
    =
    \frac{
        p_i^{(t)}\prod_{j=1}^J f_{ij,1}^{(t)}
    }{
        (1-p_i^{(t)})\prod_{j=1}^J f_{ij,0}^{(t)}
        +p_i^{(t)}\prod_{j=1}^J f_{ij,1}^{(t)}
    }.
    \label{eq:app-em-estep}
\end{equation}
The production code evaluates equation~\eqref{eq:app-em-estep} on the log
scale using a log-sum-exp calculation.

Given $\xi_i^{(t)}$, the M-step separates into standard logistic regression
problems. The latent prevalence model is updated by a fractional-response
logistic regression of $\xi_i^{(t)}$ on $Z_i$:
\begin{equation*}
    \alpha^{(t+1)}
    =
    \operatorname*{arg\,max}_\alpha
    \sum_{i=1}^n
    \left[
        \xi_i^{(t)}\log p_i(\alpha)
        +\{1-\xi_i^{(t)}\}\log\{1-p_i(\alpha)\}
    \right].
\end{equation*}
For each proxy $j$, the class-specific response models are updated by weighted
logistic regressions:
\begin{align*}
    \gamma_{j,1}^{(t+1)}
    &=
    \operatorname*{arg\,max}_\gamma
    \sum_{i=1}^n
    \xi_i^{(t)}
    \log f_j\{L_i^{(j)}\mid 1,Z_i;\gamma\},\\
    \gamma_{j,0}^{(t+1)}
    &=
    \operatorname*{arg\,max}_\gamma
    \sum_{i=1}^n
    \{1-\xi_i^{(t)}\}
    \log f_j\{L_i^{(j)}\mid 0,Z_i;\gamma\}.
\end{align*}
Thus, $\xi_i^{(t)}$ and $1-\xi_i^{(t)}$ act as the effective class-one and
class-zero weights. The M-step produces updated $p_i^{(t+1)}$ and
$\eta_{ij,a}^{(t+1)}$, which are then used in the next E-step.

The algorithm is initialized from several posterior-responsibility vectors.
The candidate starts include the row-wise proxy mean, its complement,
single-proxy starts of the form $0.1+0.8L_i^{(j)}$, and random draws from
$\operatorname{Uniform}(0.2,0.8)$ when additional starts are needed. We use six
starts and retain the solution with the largest final observed-data
log-likelihood in equation~\eqref{eq:app-em-observed-likelihood}.

Because the latent classes are identified only up to a common permutation, we
orient each fitted solution after its final EM update. All proxies are coded so that
one denotes the substantive positive class. We therefore require the average
fitted proxy contrast to be positive:
\begin{equation}
    \frac{1}{nJ}
    \sum_{i=1}^n\sum_{j=1}^J
    (\widehat\eta_{ij,1}-\widehat\eta_{ij,0})
    >0.
    \label{eq:app-em-anchor}
\end{equation}
If the fitted contrast is negative, we replace $\widehat\xi_i$ by
$1-\widehat\xi_i$ and refit the latent-prevalence and class-specific
proxy regressions. This is
the implementation counterpart of the anchoring condition in
Section~\ref{sec:latent-independent}.

\paragraph{Convergence and numerical safeguards.}
The settings used in both applications are summarized in
Table~\ref{tab:app-em-settings}. Convergence requires both a small maximum
change in posterior responsibilities and a small absolute change in the
observed log-likelihood per observation. Coefficients and
probabilities are clipped only for numerical stability.

\begin{table}[t]
    \centering
    \caption{Shared settings for the low-dimensional logistic EM implementation.}
    \label{tab:app-em-settings}
    \small
    \begin{tabular}{ll}
        \toprule
        Quantity & Setting \\
        \midrule
        Number of starts & $6$ \\
        Maximum EM iterations & $500$ \\
        Posterior-change tolerance
            & $\max_i|\xi_i^{(t)}-\xi_i^{(t-1)}|<10^{-5}$ \\
        Per-observation likelihood tolerance
            & $|\log\mathcal L(\theta^{(t)})-\log\mathcal L(\theta^{(t-1)})|/n<10^{-7}$ \\
        Logistic coefficient clipping & $[-20,20]$ \\
        Probability clipping & $[10^{-8},1-10^{-8}]$ \\
        Bridge denominator floor & $0.05$ \\
        \bottomrule
    \end{tabular}
\end{table}

The general theory is stated using cross-fitting so that flexible nuisance
learners can be accommodated without restrictive empirical-process conditions.
The application studies use the low-dimensional parametric implementation above;
for this finite-dimensional smooth nuisance class, the usual parametric
M-estimation conditions replace the role of outer cross-fitting. Specifically,
the shared EM engine is fit once on the
analysis sample in each Fowler replication or Pan--Chen bootstrap resample.
The two application adapters differ only in the construction of the proxy
matrix and $Z_i$; both call the same EM algorithm with the settings in
Table~\ref{tab:app-em-settings}. 

\subsubsection{From the EM Nuisance Fit to DMM}
\label{app:em-to-dmm}

Let
$\widehat\Delta_{ij}\coloneq\widehat\eta_{ij,1}-\widehat\eta_{ij,0}$. In the
implementation, the denominator is stabilized as
\begin{equation*}
    \widehat\Delta_{ij,\epsilon}
    =
    \begin{cases}
        \max(\widehat\Delta_{ij},\epsilon),
            & \widehat\Delta_{ij}\geq0,\\
        \min(\widehat\Delta_{ij},-\epsilon),
            & \widehat\Delta_{ij}<0,
    \end{cases}
    \qquad \epsilon=0.05,
\end{equation*}
and the fitted single-proxy bridge is
\begin{equation*}
    \widehat M_{ij}
    =
    \frac{L_i^{(j)}-\widehat\eta_{ij,0}}
         {\widehat\Delta_{ij,\epsilon}}.
\end{equation*}
The robust bridge is then
\begin{equation*}
    \widehat H_i^{\mathrm{R}}
    =
    3\binom{J}{2}^{-1}
    \sum_{j_1<j_2}\widehat M_{ij_1}\widehat M_{ij_2}
    -
    2\binom{J}{3}^{-1}
    \sum_{j_1<j_2<j_3}
    \widehat M_{ij_1}\widehat M_{ij_2}\widehat M_{ij_3}.
\end{equation*}
The floor is a finite-sample stabilization device rather than part of the
population identification argument.

For the Fowler latent-outcome analysis, let
$r_i=(1,T_i,W_i^\top)^\top$. The final estimate solves the generated-outcome
quasi-score
\begin{equation}
    \frac{1}{n}\sum_{i=1}^n
    r_i
    \left\{
        \widehat H_i^{\mathrm{R}}
        -\operatorname{expit}(r_i^\top\beta)
    \right\}
    =0.
    \label{eq:app-fowler-dmm-score}
\end{equation}
For the Pan--Chen latent-independent-variable analysis, the final estimate
solves
\begin{equation}
    \frac{1}{n}\sum_{i=1}^n
    \left[
        \{1-\widehat H_i^{\mathrm{R}}\}
        \psi_0(Y_i,W_i;\beta)
        +
        \widehat H_i^{\mathrm{R}}
        \psi_1(Y_i,W_i;\beta)
    \right]
    =0.
    \label{eq:app-pan-chen-dmm-score}
\end{equation}
Sandwich standard errors are computed from the empirical moment function
variance and Jacobian, as in equations~\eqref{eq:estimated-rs-bread}--\eqref{eq:estimated-rs-variance}.

\subsection{Monte Carlo Simulations}
\label{app:fowler-details}

\paragraph{Data, annotation task, and downstream target.}
The annotation task follows the Wesleyan Media Project codebook used in
\citet{fowler2021political}: an ad is classified according to whether its
primary purpose is to promote a candidate, attack a candidate, or contrast
candidates. We convert each LLM response into the binary indicator that the ad
is classified as Promote. The observed covariates are the Facebook indicator
$T_i$, party, incumbency status, and office type. Categorical variables are
represented by dummy indicators in both the calibration models and the EM model. The complete-case expert Promote prevalence is $0.756$.
The three highest-ranked labels are GPT-4 multi 6-shot, GPT-4 multi 3-shot,
and GPT-4 instruction 6-shot, with F1 scores of approximately $0.950$, $0.942$,
and $0.933$, respectively.

Let $Y_i^\ast$ denote the expert-coded Promote indicator and define
\begin{equation*}
    r_i=(1,T_i,W_i^\top)^\top.
\end{equation*}
The downstream working model is
\begin{equation}
    {\Pr}(Y_i^\ast=1\mid T_i,W_i)
    =
    \operatorname{expit}(r_i^\top\beta^\ast),
    \label{eq:app-fowler-latent-model}
\end{equation}
and the coefficient on $T_i$ is the target. Fitting
equation~\eqref{eq:app-fowler-latent-model} to all expert labels in the
complete-case sample gives the calibration target
$\beta_{\mathrm{FB}}^\ast=1.3141$.

\paragraph{Calibration of the data-generating process.}
We hold the $12{,}973$ observed covariate rows fixed. We first fit
equation~\eqref{eq:app-fowler-latent-model} to obtain
\begin{equation*}
    \widehat\pi_i^{\mathrm{cal}}
    =
    \operatorname{expit}(r_i^\top\widehat\beta^{\mathrm{cal}}).
\end{equation*}
For each proxy label $j$ and expert class $a\in\{0,1\}$, we separately fit
\begin{equation}
    \Pr(Y_i^{(j)}=1\mid Y_i^\ast=a,T_i,W_i)
    =
    \operatorname{expit}(r_i^\top\gamma_{j,a}),
    \label{eq:app-fowler-proxy-model}
\end{equation}
using the original expert and LLM labels. These calibration models are
ordinary logistic regressions with probability clipping for
numerical stability. Let
$\widehat\eta_{ij,a}^{\mathrm{cal}}$ denote the fitted probabilities.

In replication $b$, we draw
\begin{align*}
    Y_i^{\ast,b}
    &\sim
    \operatorname{Bernoulli}
    (\widehat\pi_i^{\mathrm{cal}}),\\
    Y_i^{(j),b}
    \mid Y_i^{\ast,b}=a,T_i,W_i
    &\sim
    \operatorname{Bernoulli}
    (\widehat\eta_{ij,a}^{\mathrm{cal}}),
    \qquad j=1,\ldots,J,
\end{align*}
with the proxy draws mutually independent across $j$ conditional on
$(Y_i^{\ast,b},T_i,W_i)$. Hence the conditional-independence restriction holds
exactly in the simulated population. Moreover, the EM estimator uses
the same logistic family as the calibration model, so its nuisance working
models are correctly specified. The feasible estimators use only
$(\widetilde Y_i^b,T_i,W_i)$; $Y_i^{\ast,b}$ is supplied only to the infeasible
oracle.

\paragraph{Proxy sets and comparison estimators.}
The binary Promote proxy labels are ranked by their F1 score against the
original expert labels. We use the nested proxy counts displayed in
Figure~\ref{fig:fowler-promote-proxy-count}. Because both the number and the
quality of the included proxies change as the set expands, this exercise
should be interpreted as adding progressively weaker proxies, not as holding
proxy quality fixed while varying $J$.

For a set of size $J$, define $S_i^b=\sum_{j=1}^JY_i^{(j),b}$. The
naive majority-vote outcome is
\begin{equation*}
    \check Y_i^b
    =
    \mathbf{1}\{S_i^b>J/2\}
    +
    B_i^b\mathbf{1}\{S_i^b=J/2\},
    \qquad
    B_i^b\sim\operatorname{Bernoulli}(1/2),
\end{equation*}
where the auxiliary tie-breaking draws are independent across tied samples and
Monte Carlo replications. The naive estimator fits the same logistic
regression as equation~\eqref{eq:app-fowler-latent-model} after replacing
$Y_i^{\ast,b}$ with $\check Y_i^b$. DMM uses the same $J$ proxies but constructs
$\widehat H_i^{\mathrm{R}}$ based on nuisance components fitted with the aforementioned EM algorithm and solves
equation~\eqref{eq:app-fowler-dmm-score}. The oracle fits the downstream logit
using $Y_i^{\ast,b}$. Each proxy-count cell uses $500$ Monte Carlo
replications.

For an estimator $m$, the reported summaries are
\begin{align*}
    \operatorname{Bias}(m)
    &=
    \frac{1}{500}\sum_{b=1}^{500}
    (\widehat\beta_{m,b}-\beta_{\mathrm{FB}}^\ast),\\
    \operatorname{RMSE}(m)
    &=
    \left\{
    \frac{1}{500}\sum_{b=1}^{500}
    (\widehat\beta_{m,b}-\beta_{\mathrm{FB}}^\ast)^2
    \right\}^{1/2},\\
    \operatorname{Coverage}(m)
    &=
    \frac{1}{500}\sum_{b=1}^{500}
    \mathbf{1}\{
        \beta_{\mathrm{FB}}^\ast
        \in
        \widehat C_{m,b}^{95\%}
    \}.
\end{align*}

\paragraph{Additional results.}
The three expert-coded tone indicators differ substantially in prevalence:
Promote accounts for $75.6\%$ of advertisements in the complete-case sample,
whereas Contrast accounts for $17.2\%$ and Attack for only $7.1\%$. We repeat
the same application-calibrated simulation separately for the two less
prevalent indicators, recalibrating the latent-outcome and proxy models for
each outcome. The target Facebook coefficients are $-1.029$ for Contrast and
$-1.334$ for Attack. Table~\ref{tab:fowler-alternative-tone-results} shows that
DMM performs especially well for Attack: across the top-3 through top-7
sets, its bias remains close to zero, coverage ranges from $0.942$ to
$0.962$, and its RMSE falls from $0.113$ to $0.085$, approaching the oracle
RMSE of $0.076$. The corresponding naive estimators are substantially biased
and severely undercover. Contrast shows similar results.

\begin{table}[htbp]
    \centering
    \caption{Fowler simulation results for the alternative Attack and Contrast
    tone indicators. Each cell uses $500$ Monte Carlo replications. For naive
    rows, F1 is the individual proxy's positive-class F1 score; for DMM rows,
    it is the median [minimum, maximum] F1 score among the included proxies.}
    \label{tab:fowler-alternative-tone-results}
    \footnotesize
    \setlength{\tabcolsep}{3.5pt}
    \begin{tabular}{@{}llcrrr@{}}
        \toprule
        Outcome & Estimator / set & F1 summary & Bias & Coverage & RMSE \\
        \midrule
        Attack
        & Oracle & --- & $-0.002$ & $0.954$ & $0.076$ \\
        & Naive: GPT-4 multi, 3-shot
        & $0.651$ & $0.591$ & $0.000$ & $0.594$ \\
        & Naive: GPT-4 multi, 0-shot
        & $0.595$ & $0.890$ & $0.000$ & $0.892$ \\
        & Naive: Llama-2 multi, 6-shot
        & $0.581$ & $0.113$ & $0.410$ & $0.124$ \\
        & DMM: Top 3
        & $0.595$ [$0.581$, $0.651$] & $-0.005$ & $0.958$ & $0.113$ \\
        & DMM: Top 5
        & $0.581$ [$0.565$, $0.651$] & $-0.001$ & $0.962$ & $0.087$ \\
        & DMM: Top 7
        & $0.577$ [$0.548$, $0.651$] & $-0.001$ & $0.942$ & $0.085$ \\
        \addlinespace
        Contrast
        & Oracle & --- & $0.000$ & $0.952$ & $0.050$ \\
        & Naive: GPT-4 multi, 3-shot
        & $0.665$ & $-0.275$ & $0.000$ & $0.280$ \\
        & Naive: GPT-4 multi, 0-shot
        & $0.658$ & $-0.308$ & $0.000$ & $0.312$ \\
        & Naive: Llama-2 instruction, 6-shot
        & $0.511$ & $-0.144$ & $0.340$ & $0.155$ \\
        & DMM: Top 3
        & $0.658$ [$0.511$, $0.665$] & $-0.083$ & $0.958$ & $0.598$ \\
        & DMM: Top 5
        & $0.511$ [$0.489$, $0.665$] & $0.005$ & $0.950$ & $0.155$ \\
        & DMM: Top 7
        & $0.494$ [$0.375$, $0.665$] & $0.003$ & $0.960$ & $0.163$ \\
        \bottomrule
    \end{tabular}
\end{table}

\subsection{Empirical Validation}
\label{app:pan-chen-details}

\paragraph{Data and downstream model.}
The analysis uses $1{,}412$ Chinese citizen complaints. The latent independent
variable is the expert-coded indicator
$X_i^\ast=\mathbf{1}\{\text{prefecture-level wrongdoing}\}$, and the downstream
outcome $Y_i=\texttt{SendOrNot}_i$ indicates whether the complaint was sent to
higher-level authorities. The expert-coded positive-class prevalence is
$0.055$, and the mean of $Y_i$ is $0.418$.

The downstream model is
\begin{equation}
    \operatorname{logit}
    \{{\Pr}(Y_i=1\mid X_i^\ast,W_i)\}
    =
    \alpha+\tau X_i^\ast+\gamma^\top W_i,
    \label{eq:app-pan-chen-logit}
\end{equation}
where $W_i$ contains
\texttt{connect2b}, \texttt{prevalence}, \texttt{regionj},
\texttt{groupIssue}, \texttt{realWorldCollectiveAction},
\texttt{petitioning}, \texttt{sentiment\_indico}, and
\texttt{personal\_experience}. Seven controls are binary and
\texttt{sentiment\_indico} is continuous. The coefficient $\tau$ is the main
estimand. The full-sample expert-label fit gives
$\widehat\tau_{\mathrm{expert}}=-1.0388$.

\paragraph{Proxy labels.}
We use three LLM labels:
GPT-4.1 5-shot, GPT-4 5-shot, and Llama-4 0-shot. Table~\ref{tab:app-pc-proxy-quality}
reports several descriptive metrics against the expert labels.

\begin{table}[t]
    \centering
    \caption{Descriptive quality of the three Pan--Chen proxy labels. The
    class-weighted F1 is the quantity displayed in
    Figure~\ref{fig:pan-chen-empirical-results}; positive-class F1 and
    sensitivity focus on the rare wrongdoing class.}
    \label{tab:app-pc-proxy-quality}
    \small
    \begin{tabular}{lccccc}
        \toprule
        Proxy & Weighted F1 & Positive F1 & Accuracy & Sensitivity & Specificity \\
        \midrule
        GPT-4.1 5-shot & 0.944 & 0.545 & 0.938 & 0.667 & 0.954 \\
        GPT-4 5-shot   & 0.942 & 0.511 & 0.939 & 0.577 & 0.960 \\
        Llama-4 0-shot & 0.949 & 0.492 & 0.953 & 0.410 & 0.985 \\
        \bottomrule
    \end{tabular}
\end{table}

\paragraph{DMM and benchmark estimators.}
The expert-label benchmark fits equation~\eqref{eq:app-pan-chen-logit} using
$X_i^\ast$. Each naive estimator replaces $X_i^\ast$ with one LLM label and
otherwise fits the same model. DMM estimates
nuisance components using the aforementioned EM algorithm, constructs the robust bridge, and solves
equation~\eqref{eq:app-pan-chen-dmm-score}.

The DSL benchmark uses a simple random sample of $500$ expert labels, or about
$35.4\%$ of the sample. It supplies the three LLM labels jointly, together with
$W_i$, to the supervised prediction step. In the implementation,
\texttt{dsl::dsl} uses the logit downstream model and its default generalized
random forest learner with cross-fitting. Thus, DSL neither selects one proxy
nor converts the three proxies to a majority vote; it learns a joint predictor
of $X_i^\ast$ and then applies the design-based correction using the sampled
expert labels.

\paragraph{Fixed-target bootstrap diagnostic.}
We estimate empirical coverage by asking how often each method's
interval contains the fixed full-sample expert-label coefficient under
nonparametric resampling of the observed complaints.

For bootstrap draw $b$, we sample $n=1{,}412$ complaints with replacement,
refit the complete estimator, and record the interval
$[\widehat\tau_{b,\mathrm{lo}},\widehat\tau_{b,\mathrm{hi}}]$. The displayed
fixed-target rate is
\begin{equation*}
    \frac{1}{500}\sum_{b=1}^{500}
    \mathbf{1}\{
        \widehat\tau_{b,\mathrm{lo}}
        \leq
        \widehat\tau_{\mathrm{expert}}
        \leq
        \widehat\tau_{b,\mathrm{hi}}
    \},
    \qquad
    \widehat\tau_{\mathrm{expert}}=-1.0388.
\end{equation*}

\end{document}